\documentclass[11pt]{article}
\usepackage[onehalfspacing]{setspace}
\usepackage{graphicx} 
\usepackage{natbib}
\usepackage[margin=1in]{geometry}
\usepackage{amsmath}
\usepackage{amsthm}
\usepackage{pifont}
\usepackage{subcaption}
\usepackage{xcolor}
\usepackage{ulem}
\usepackage{tikz}
\usetikzlibrary{arrows.meta, positioning}
\usetikzlibrary{shapes.geometric}
\usetikzlibrary{calc}
\usepackage{booktabs}
\usepackage{tabu}
\usepackage[T1]{fontenc}
\usepackage{amssymb}
\usepackage[colorlinks=true, allcolors=blue, breaklinks=true]{hyperref}
\usepackage{algorithm}
\usepackage{algorithmic}
\usepackage{url}

\usepackage{pgfplots}
\usepgfplotslibrary{groupplots}
\pgfplotsset{compat=1.18}
\usepackage[toc]{appendix}
\usepackage{titlesec}
\usepackage{cancel}

\theoremstyle{plain} 
\newtheorem{assumption}{Assumption}
\newtheorem{definition}{Definition}
\newtheorem{theorem}{Theorem}
\newtheorem{lemma}{Lemma}
\newtheorem{proposition}{Proposition}
\newtheorem{corollary}{Corollary}
\theoremstyle{definition} 
\newtheorem{example}{Example}

\newtheorem{remark}{Remark}

\begin{document}

\title{Conjectural Variations in Competitive Dynamic Pricing: A Learning Foundation via Experimentation Design and Feedback Structure}

\author{Bar Light\protect\thanks{Business School and Institute of Operations Research and Analytics, National University of Singapore, Singapore. e-mail: \textsf{barlight@nus.edu.sg} }   \and Wenyu Wang\protect\thanks{Institute of Operations Research and Analytics, National University of Singapore, Singapore. e-mail: \textsf{wang.wenyu@u.nus.edu} }}


\maketitle

\maketitle


\begin{abstract}
We study competitive dynamic pricing among multiple sellers, motivated by the
rise of large-scale experimentation and algorithmic pricing in retail and online
marketplaces. Sellers repeatedly set prices using simple learning rules and
observe their own realized demand, while possibly observing only a subset of
rivals' prices, even though demand depends on all sellers' prices and is subject
to random shocks. Each seller runs local price experiments, such as switchback-style designs, and
updates a focal price using a linear demand estimate fitted to its own demand
data and the competitor prices it observes.  Under
certain conditions on demand, the resulting dynamics converge to a Conjectural
Variations (CV) equilibrium, a classic static equilibrium notion in which each
seller best responds under a conjecture that rivals' prices co-move 
systematically to changes in its own price.
Unlike standard CV models that treat conjectures as behavioral primitives, we
show that these conjectures arise endogenously from the interaction between the
feedback structure and the correlation structure of experimentation. When a
seller does not observe some rivals' prices, correlated experimentation induces
an omitted-variable bias in demand estimation. We show that this bias determines the
conjectures that govern the long-run equilibrium. Notably, when this learning
bias vanishes, for example under full price feedback or independent
experimentation of unobserved rivals, the learning dynamics converge to the
standard Nash equilibrium. We provide simple sufficient conditions on demand for
convergence in standard models and establish a finite-sample guarantee, showing that the mean squared price error decays at a rate of  $\widetilde O (T^{-1/2})$.
Our results imply that, in competitive markets, experimentation design and
feedback structure jointly determine the equilibrium prices reached by practical
learning algorithms.
\end{abstract}

\newpage

\section{Introduction}
The digitization of commerce, the growth of online marketplaces, and rapid improvements in technology have made it feasible to adjust prices at high frequency and to measure demand responses at scale, accelerating the use of experimentation and algorithmic pricing in practice. Recent papers document large-scale randomized pricing interventions, including discounts and promotions, e.g., a randomized field experiment on Alibaba involving more than 100 million customers \cite{Zhang2020AlibabaPromotions}, as well as a pricing meta-experiment on Airbnb \cite{Holtz2025AirbnbMetaExperiment}, just to name a few. In addition to large-scale experimentation by platforms themselves, a substantial number of third-party sellers deploy algorithmic pricing  (e.g., \cite{chen2016empirical,calzolari2025pricing}).

These developments have motivated a growing operations literature on dynamic pricing with demand learning through experimentation and the practical constraints such experimentation faces in the field. Yet most of the theoretical analysis focuses on a single seller learning demand in isolation, abstracting away from the competitive environments in which experimentation is often deployed \cite{BesbesZeevi2009,Cheung2017LimitedExperimentation}. Moving from a single seller to competition is challenging because learning and experimentation interact across firms: when multiple competitors experiment at the same time, each firm’s measured treatment effects can be distorted by rivals’ contemporaneous actions, creating bias in inference and decision-making (see \cite{Waisman2025ParallelExperimentation} for discussion in the context of online advertising). At the same time, a rapidly growing simulation and computational literature shows that competitive learning dynamics can generate supra-competitive outcomes. Despite these insights,  theoretical results characterizing the long-run outcomes of practical pricing algorithms in competitive markets remain limited. This paper takes a step toward bridging this gap by analyzing a simple dynamic pricing algorithm and proving that its learning dynamics converge to a static equilibrium. Moreover, the limiting equilibrium is pinned down by the learning bias induced by correlated experimentation and the feedback structure in the market.

Specifically, we study a repeated dynamic pricing game with multiple sellers. In each period, each seller posts a price using a pricing algorithm, and demand depends on the full vector of sellers' prices and is subject to random shocks. Sellers observe their own prices and realized demands, but they may differ in the competitor-price information available to them. In particular, each seller may observe the prices of only a subset of rivals, and sellers do not observe rivals' demands or know the true demand system. This feedback structure captures the fact that different markets can have very different information environments. In some markets, sellers may have very limited information and rely only on their own prices and demands; this corresponds to the bandit-feedback case. In other markets, sellers may be able to observe all rival prices; this corresponds to full price feedback. Between these two extremes, sellers may observe only a subset of competitors, such as nearby competitors, sellers on the same platform, prominent or dominant sellers, or rivals whose prices are easier to track. The feedback structure can also be asymmetric: one seller may observe another seller's price even if the reverse is not true.

This flexible feedback structure is motivated by practice. Even when sellers can track some competitors' posted prices, such information may not resolve the learning problem: sellers may not know whether observed prices reflect experimentation or routine pricing, or whether the observed rivals are the competitors most relevant for demand. For example, an Amazon seller may compete not only with other Amazon sellers, but also with sellers on other e-commerce platforms and with local brick-and-mortar stores, making comprehensive monitoring difficult. We therefore allow each seller to observe an arbitrary subset of rival prices, nesting both bandit feedback and full price feedback as special cases.

Given this information environment, we model sellers as learning through local price experimentation around a current focal price. Experimentation is organized in batches: during a batch, each seller repeatedly perturbs her focal price, observes demand and the rival prices available under her feedback structure, and then updates the focal price using the batch data. The perturbation may be binary, corresponding to a two-point A/B or switchback design that alternates between control and treatment prices; it may also be three-point, corresponding to low, baseline, and high experimental prices; or it can be a general random variable.

Crucially, our framework allows experimentation to be statistically dependent across sellers. 
While sellers typically do not explicitly coordinate their experiments, such statistical dependence can arise endogenously from the structural features of modern marketplaces. 
First, technical constraints such as platform-imposed API throttling can make it infeasible to query and update prices continuously, encouraging automated pricing systems to run repricing and experimentation in periodic batch jobs that are naturally aligned to common wall-clock boundaries across sellers \cite{chen2016empirical}.\footnote{Amazon's Selling Partner API is subject to request throttling via token-bucket rate limits \cite{AmazonSPAPIRateLimits}. Amazon has also announced an annual subscription fee (effective January 31, 2026) and monthly usage-based fees for third-party developers offering applications to other selling partners, which can further increase incentives for call optimization and batching \cite{AmazonSPAPIFeesAnnouncement}.} Even sellers that would prefer to randomize the timing of experiments may therefore be constrained to coarser update windows. 
When experiments are implemented at these synchronized decision epochs, many sellers’ price perturbations can occur contemporaneously, creating statistical dependence in the realized price variation even without explicit coordination. 
Second, reliance on third-party algorithmic pricing services can introduce additional dependence when shared software infrastructure pushes price updates to many clients on common schedules (or within common processing windows). Consistent with this, \cite{musolff2025} documents that repricing tools on Amazon frequently implement deterministic ``resetting" behavior at night (e.g., concentrated in low-sales nighttime hours) and discusses evidence that a repricer can become overwhelmed by scheduled repricing requests, both pointing to coarse, synchronized timing of price changes across many sellers. 
Third, correlation can also arise when common market-level demand shifts lead many sellers’ pricing systems to reset estimates or increase exploration in subsequent periods, as can occur in adaptive pricing systems. Finally, platform-initiated mechanisms, e.g., widespread price reductions during major holiday events like 'Black Friday' or '11.11', can introduce common variation in prices, potentially synchronizing the price perturbations observed across competitors (see \cite{Zhang2020AlibabaPromotions} for a related massive price promotion on Alibaba).

Using the batch observations generated by local price experimentation, each seller estimates a linear demand model for her own demand using only the information available under her feedback structure: her own prices and demand observations, together with the subset of rival prices she observes. She then computes the own-price revenue-maximizing price implied by this fitted linear demand model, holding the observed rival prices at their batch averages, and updates her focal price by moving partially toward that target, subject to feasibility bounds. Because the fitted coefficients are estimated from finitely many noisy demand observations, each update is a noisy, data-driven approximation to the population fitted-model revenue maximizer.

Our analysis relates these multi-seller dynamic pricing learning dynamics to a classic equilibrium concept from industrial organization: the Conjectural Variations (CV) equilibrium. The defining feature of a CV equilibrium is that each seller chooses a price that is optimal not only given current rivals' prices, but also given a conjecture about how rivals' prices co-move with changes in her own price. These conjectures can be summarized by a matrix $A$. The entry $A_{ij}$ represents seller $i$'s conjectured marginal co-movement of seller $j$'s price with seller $i$'s price: if $A_{ij}>0$, seller $i$ acts as if seller $j$'s price increases when $i$ raises her price; if $A_{ij}<0$, seller $i$ acts as if seller $j$'s price decreases when $i$ raises her price; and if $A_{ij}=0$, seller $i$ acts as if seller $j$'s price does not systematically co-move with changes in $i$'s price. Consequently, when $A_{ij}=0$ for all $j\neq i$, each seller behaves as if rivals' prices are locally fixed, and the CV equilibrium reduces to the standard Nash equilibrium.

Our main theoretical result shows that, under certain conditions on demand, when each seller follows the partial-feedback pricing procedure described above, the resulting price sequence converges to a CV equilibrium, even though sellers do not know the true demand system and do not explicitly reason about competitors' reactions. Moreover, the limiting conjecture matrix is not imposed as an exogenous behavioral primitive; instead, it is endogenously determined by the interaction between the feedback structure and the correlation structure of experimentation.

When experimentation is correlated across sellers, the price variation observed by a seller may systematically co-occur with the perturbations of rivals whose prices she does not observe. Because these unobserved rival prices are omitted from the seller's demand regression, her demand estimation effectively treats their contemporaneous price movements as unobserved covariates. This generates a systematic omitted-variable-type bias: some of the demand change driven by unobserved competitors' price movements is mistakenly attributed to the prices included in the seller's fitted demand model, including her own price. We prove that, as the algorithm learns from an increasing amount of data and experimentation magnitudes shrink, the resulting biased demand estimates lead sellers to behave as if unobserved competitors' prices co-move with their own according to a particular conjecture matrix, and the limiting prices coincide with the corresponding CV equilibrium. In this sense, correlated experimentation and partial feedback provide a learning foundation for conjectural variations: the same statistical dependence that generates bias in demand learning also determines the conjectures that characterize the static equilibrium approached by the learning dynamics. Interestingly, when this learning bias vanishes, for example under full price feedback or when unobserved rivals' experimentation is asymptotically uncorrelated with the observed price variation, the induced conjectures are zero and the same learning dynamics converge to the Nash equilibrium of the underlying static pricing game.

We show that the economic mechanism behind this result is broader than the particular linear demand estimate used in the algorithm. The linear regression step  lets us explicitly characterize the induced conjecture matrix in terms of partial linear projections determined by the feedback and experimentation structures. However, the CV mechanism itself comes from a more basic object: the demand relationship a seller can learn from the data she actually observes. When some rival prices are unobserved, this learned demand relationship averages over those omitted prices conditional on the observed price variation. As a result, other learning rules that estimate marginal demand from the observed experimental variation can inherit the same conjectural-variation component. In Section~\ref{sec:intuition-beyond-ldl}, we discuss this in detail. Further, in Section~\ref{sec:beyond-contraction}, we show that even when the sufficient stability conditions for convergence fail, the same conjecture matrix continues to govern the direction of the limiting price adjustments. Thus, more generally, the feedback and experimentation structures determine the price dynamics followed by the learning process, not only the CV equilibrium reached under the convergence conditions.

In Section~\ref{sec:supra_competitive}, we study how this mechanism affects prices. 
Because feedback and experimentation determine the induced conjecture matrix, comparative statics in conjectures translate directly into comparative statics in market structure and experimental design. Under natural strategic-complementarity conditions, we show that equilibrium prices are increasing in the conjectures. Thus, positive induced conjectures can lead the learning dynamics to select prices above the Nash benchmark; under bandit feedback, this occurs transparently when experimentation is positively correlated across sellers. The broader partial-feedback case is more nuanced. Since conjectures are partial projection coefficients, their sign and magnitude depend on which rival prices are observed. We show, for example, that in a market with a major seller observed by smaller sellers, positive pairwise experimentation correlations can become negative induced conjectures, leading to prices below Nash. We also discuss in the same section why we take feedback and experimentation structures as given rather than endogenizing them from a seller or platform perspective. Such a model would require sellers or the platform to know detailed information about demand primitives, rivals' experimentation rules, and competitors' feedback structures. Moreover, we provide examples that show that the effects of observing more rivals or changing the correlation structure are market-specific and are non-monotone for individual seller revenues.

In Section~\ref{sec:converge_rate}, we establish an explicit finite-sample convergence guarantee. Subject to the conditions that imply convergence to the feedback- and experimentation-induced CV equilibrium, the mean squared price error decays on the order of $T^{-1/2}$ up to logarithmic factors. This is the same order as the optimal mean-squared rate established for related Nash-convergence bandit-feedback games \cite{ba2025doubly}. Hence, simple learning algorithms based on practical price experimentation and misspecified local linear demand models can achieve convergence guarantees that match the best-known rates in related settings, even when the limiting equilibrium is a nonzero-CV equilibrium rather than Nash.

In Section~\ref{sec:examples}, we provide conditions on demand for convergence of the learning dynamics to a CV equilibrium. The intuition is that the fitted local pricing target must not move too sharply when current prices change. This can fail for two reasons: competitive and conjectural effects may be too strong relative to own-price effects, or demand may be too curved for the local linear approximation. We formalize these forces through demand derivatives, separating first-order competitive and conjectural effects from curvature effects that reflect changes in the local demand slope. This yields simple criteria in standard models. For linear demand,  the condition takes a form similar to standard diagonal dominance. For multinomial logit  demand (MNL), stability  typically holds when market shares are not too large.

Taken together, our results identify the joint structure of price feedback and experimentation as a determinant of equilibrium selection. Correlated experimentation with omitted rival prices induces a systematic learning bias, captured by an endogenous conjecture matrix, that can move the market away from Nash toward a CV equilibrium. Under strategic complementarities, positive induced conjectures raise prices relative to Nash, even though firms neither communicate nor explicitly coordinate; other feedback structures can generate zero or negative conjectures, leading to Nash or below-Nash outcomes. Thus, experimentation and feedback affect markets not only by improving individual demand estimates, but also by shaping the price dynamics and equilibrium selected by decentralized pricing algorithms.

\subsection{Related Literature}

\textbf{Demand Learning with Price Experimentation in a Monopoly Setting.} Dynamic pricing problems under demand uncertainty have been extensively studied in the revenue management literature. A substantial body of work uses price experimentation to trade off between exploration (setting prices to gather information about demand) and exploitation (setting prices to maximize current revenue based on existing beliefs). Several studies  analyze the performance of Bayesian pricing policies \cite{farias2010dynamic,harrison2012bayesian}. \cite{keskin2014dynamic} establishes sufficient conditions for asymptotic optimality under linear multi-product demand models, while \cite{broder2012dynamic} studies dynamic pricing under a general parametric choice model. \cite{den2014simultaneously} adopts a statistical perspective and proposes controlled variance pricing. More recent work incorporates contextual information under linear and generalized linear demand learning frameworks \cite{nambiar2019dynamic,ban2021personalized,wang2025dynamic}. From an experimental design viewpoint, \cite{bastani2022meta,simchi-levi2025pricing} study how experimentation shapes learning-vs-earning in linear demand models. Most of these works focus on settings with a (partially) known demand structure but unknown parameters. On the other hand, \cite{BesbesZeevi2009} adopts a deterministic experimentation scheme to analyze both parametric and nonparametric demand models, and shows that “learning on the fly” can be asymptotically efficient. Building on this framework, \cite{besbes2015surprising} studies a  deterministic price experimentation algorithm based on scheduled, nonrandom perturbations around a focal price, and uses it to show that potentially misspecified demand models can perform well. \cite{bu2025context} studies contextual dynamic pricing under an additive separable demand model with unknown components. \cite{Cheung2017LimitedExperimentation, perakis2024dynamic} also study settings with an unknown demand structure, but focus on regimes with limited price experimentation. 

A fundamental distinction between our work and this stream of literature is that these studies abstract away from competitive effects by focusing on a single seller. In contrast, our analysis centers on the competitive dynamics that arise when multiple sellers run their own experimentation schemes and pricing algorithms  concurrently.

\textbf{Dynamic Pricing in Competitive Environments.} While monopoly dynamic pricing is well studied, comparatively fewer papers analyze learning and pricing in competitive pricing environments. One stream of work focuses on how the pricing game can lead to Nash equilibrium.
\cite{cooper2015learning} studies a repeated duopoly in which each seller estimates a monopoly-style demand model without price experimentation. They identify settings in which the induced learning dynamics converge to the Nash equilibrium, to the cooperative solution, or to other steady state prices that are neither and depend on initial conditions. 
\cite{birge2024interfere} investigates competitive learning from a platform perspective, studying when the platform should intervene in information revelation to sellers.  Relatedly, \cite{banerjee2024price} study competitive pricing under
platform-induced consideration sets and establish conditions under which
distributed gradient-based pricing dynamics converge to Nash equilibria. 
\cite{li2024lego} studies sequential price competition under unknown linear demand and proposes a gradient-optimization policy that converges to the Nash equilibrium. \cite{li2025adaptive} studies a broader uncertain sequential-competition framework, identifying conditions under which learning from individual data leads to Nash equilibrium.    \cite{meylahn2022learning} and \cite{LootsDenBoer2023} construct demand-learning pricing algorithms that use price perturbations to support collusive or supra-competitive outcomes under self-play, while reverting to competitive or best-response behavior against noncollusive opponents. 
\cite{yang2024competitive} proposes a noncooperative pricing algorithm with coordinated price experimentation, in which firms adjust prices according to a prescribed schedule so that prices converge to the Nash equilibrium. In contrast, our setting is decentralized. 
\cite{bracale2025revenue} considers $s$-concave demand with known concavity parameter and proposes an algorithm that guarantees convergence to the Nash equilibrium.

Another growing body of work studies how algorithms can facilitate collusion in market settings. Many of these papers use simulations or empirical studies to show that algorithmic pricing can learn to sustain supra-competitive outcomes, highlighting the potential for tacit collusion even without explicit communication. See \cite{Calvano2020AER,Hansen2021AlgorithmicCollusion,deng2024algorithmiccollusiondynamicpricing,asker2022artificial,musolff2025} for related work on algorithmic collusion.
While much of this literature relies on simulation, a smaller stream of work provides theoretical results. \cite{banchio2023artificial} develops a dynamical-systems approximation of multi-agent Q-learning and identifies “spontaneous coupling,” an endogenous linkage in independent learners’ value estimates that can sustain recurrent supra-competitive phases under low exploration.  One recent paper \cite{bichler2025onlineoptimizationalgorithmsrepeated} proves that, in finite normal games, under a broad class of online optimization algorithms, sellers' actions converge to a subset of rationalizable outcomes. They demonstrate that heterogeneity in sellers' algorithms may eliminate supra-competitive prices. 
 Unlike our demand-learning setting, much of this literature studies bandit-style adaptive exploration, rather than estimating demand through designed price experimentation.

Two independent and contemporaneous papers are closest to our work. 
\cite{lin2025competition} identifies correlated exploration as one mechanism through which misspecified pricing algorithms generate supra-competitive prices. 
\cite{yang2026driven} focuses on price imitation as a mechanism that can push prices toward collusive outcomes. 
Both papers focus on symmetric linear-demand environments.
In contrast, our framework allows for more general demand systems such as MNL, arbitrary partial-feedback structures, and a broad class of experimentation schemes. 
To the best of our knowledge, ours is the first paper to show that feedback and experimentation structures jointly induce an endogenous conjectural-variations equilibrium. 
This perspective also shows that non-Nash outcomes need not be supra-competitive, even with correlated experimentation: under partial feedback, the induced conjectures can be negative and prices may fall below the Nash benchmark.

\textbf{Conjectural Variations (CV) and CV Equilibria. } The analysis of players’ behavior in oligopoly games has a long history. An early work \cite{bowley1924mathematical} introduced the concept of conjectural variations (CV), which provided a static benchmark to think about competitors' behaviors. Since its introduction, the concept has been explored and discussed in the academic literature for several decades (see \cite{perry1982oligopoly,figuieres2004theory,vives1999oligopoly}).
A common critique of CV equilibria is that conjectures are interpreted as psychological conjectures about rivals’ responses, which need not coincide with the true strategic environment. To address this, \cite{bresnahan1981duopoly} formally introduced the concept of consistent conjectures, requiring that players' conjectures coincide with the actual slope of reaction functions in equilibrium. Even with this criticism, \cite{cabral1995conjectural} argues that the CV can be regarded as a reduced-form of dynamic games. Relatedly, \cite{brown2023competition} shows that pricing algorithms can induce conjectural-variations-type outcomes through commitment to pricing rules that react to rivals' prices. Our mechanism is different: conjectures arise from biased demand learning induced by feedback and experimentation. Recently, \cite{calderone2023consistent} integrates the concept of (consistent) CV into the study of control and learning systems. They analyze the local stability properties of the dynamics around (consistent) CV equilibria.

Our paper revisits this classic equilibrium concept without imposing exogenous or psychological conjectures on sellers. In our setting, sellers need not even be aware of all of their competitors; they simply maximize revenue against a misspecified demand model. We show that the CV equilibrium arises endogenously from the statistical structure of feedback and experimentation itself. This provides a learning-based foundation for Conjectural Variations, bridging the gap between an abstract ``mental" solution concept and a realizable outcome of learning algorithms.

\section{Main Convergence Result}

\subsection{The Dynamic Pricing Game} \label{sec:model}

We consider an $n$-seller repeated pricing game over periods $t=1,2,\ldots$. In each period $t$, sellers simultaneously post prices $\mathbf{p}^t=(p_1^t,\ldots,p_n^t)^\top$, where seller $i$ chooses $p_i^t$ from a compact interval $\mathcal{P}_i=[p_i^{l},p_i^{h}]\subset(0,\infty)$. We interpret $\mathcal{P}_i$ as the economically relevant price range in which seller $i$ experiments and in which the learning dynamics and candidate long-run outcomes (e.g., equilibrium prices) are sought. Such bounds are standard in practice, e.g., sellers typically impose minimum and maximum prices to respect cost and margin-based floors and to comply with marketplace repricing constraints. 
  Let $\mathcal{P}=\prod_{i=1}^n \mathcal{P}_i$ and for any subset \(J\subseteq[n]\), let
\(\mathcal P_J:=\prod_{j\in J}\mathcal P_j\). Denote $\mathbf{p}^l=(p_1^{l},\ldots,p_n^{l})$ and $\mathbf{p}^h=(p_1^{h},\ldots,p_n^{h})$ as the componentwise lower and upper feasible prices, and as usual, denote by $\mathbf{p}_{-i}$ the vector of prices of all sellers except $i$.

  Given the price vector $\mathbf{p}^t$ in period $t$, seller $i$'s realized demand is $D_i^t = \lambda_i(\mathbf{p}^t) + \varepsilon_i^t$, where $\lambda_i:\mathcal{P}\rightarrow\mathbb{R}_+$ is a deterministic, twice continuously differentiable mean demand function, and the demand shock
vectors $\{\varepsilon^t\}_{t\ge1}$ are non-degenerate and i.i.d. over time and exogenous to the sellers' pricing process: conditional on the past,
the shocks in a period are independent of any randomization used to choose
prices in that period. For each seller $i$, the demand shocks have zero mean $\mathbb E[\varepsilon_i^t]=0$ and finite variance
$\mathrm{Var}(\varepsilon_i^t)=\sigma_i^2\in(0,\infty)$.  Further, we assume that there exists a constant $s_0>0$ such that, for all sellers $i\in[n]$ and all $|s|\le s_0$, 
$\mathbb{E}[\exp\{s\varepsilon_i^t\}]<\infty$. We allow arbitrary
contemporaneous correlation among $(\varepsilon_1^t,\ldots,\varepsilon_n^t)$ that can capture common shocks between sellers.
To avoid negative demand, we assume that realized demand is non-negative almost surely\footnote{Nonnegativity is a modeling convenience and is not essential for our analysis; the results extend to other specifications under the stated tail/moment condition.}
 (i.e., $\lambda_i(\mathbf{p}^t)+\varepsilon_i^t\ge0$ a.s.) for all $i\in[n]$ and $\mathbf{p}^t\in\mathcal{P}$, and that demand satisfies standard regularity conditions: for any $\mathbf{p}\in\mathcal{P}$, $\partial_{p_i}\lambda_i(\mathbf{p})<0 \text{ and } \partial_{p_j}\lambda_i(\mathbf{p})>0$ for all $j\ne i$ (i.e., demand is decreasing in seller $i$'s own price and exhibits positive cross-price effects with respect to rivals' prices).

In period $t$, given prices $\mathbf{p}^t=(p_i^t,\mathbf{p}_{-i}^t)$, seller $i$'s expected instantaneous revenue is
$
r_i(\mathbf{p}^t) \;=\; p_i^t\,\lambda_i(\mathbf{p}^t).
$

\subsection{Feedback and Experimentation Structures}\label{sec:corr_feedback}
We now introduce two key features of our model: the feedback
structure, which determines which rival prices each seller observes, and the
experimentation structure, whose joint distribution determines the correlation
among sellers' price perturbations.

\textbf{Feedback structure.} We allow sellers to differ in the price information they observe. For each seller $i$, let $\mathcal S_i\subseteq[n]\setminus\{i\}$ denote the set of rivals whose prices seller $i$ observes. Define $\mathcal R_i:=\{i\}\cup\mathcal S_i$ and $\mathcal U_i:=[n]\setminus\mathcal R_i$. Thus $\mathcal R_i$ is the set of prices included in seller $i$'s learning procedure, while $\mathcal U_i$ is the set of rivals whose prices are not observed by seller $i$. Let $\mathcal G:=(\mathcal S_1,\ldots,\mathcal S_n)$ denote the feedback structure. Thus, seller $i$ always observes her own price $p_i^t$ and realized demand $D_i^t$, and in addition observes the rival prices $\{p_j^t:j\in\mathcal S_i\}$. Sellers do not observe rivals' demands $D_j^t$, $j\ne i$, and they do not know the true demand function $\lambda_i(\cdot)$.

This directed feedback structure captures a range of market configurations. In markets with very limited information, $\mathcal S_i=\emptyset$ for all sellers, so each seller has only bandit feedback and learns from her own prices and demands. In markets with full price feedback, $\mathcal S_i=[n]\setminus\{i\}$ for all sellers, so every seller observes every rival's price. Between these extremes, sellers may observe only a subset of competitors. For example, small sellers may track the prices of a prominent or dominant seller, while the dominant seller may not track every small seller; alternatively, sellers may monitor only nearby competitors, sellers using the same platform, or a small set of salient rivals. The sets $\mathcal S_i$ need not be symmetric: seller $i$ may observe seller $j$ even if seller $j$ does not observe seller $i$.

\textbf{Experimentation structure.}
In our model, experimentation and learning take place in batches. Sellers do not
re-estimate demand and update their focal price after every individual period.
Instead, they collect demand observations over a batch of pricing periods,
experiment locally around a current focal price during that batch, and update
the focal price only after the batch is completed. This timing is motivated by practical implementations of pricing analytics in
which experiments, repricing, and model recalibrations are carried out over
discrete update windows rather than continuously, so that enough experimental
price-demand observations are collected before the fitted demand model is
updated. 
The specific estimation and reoptimization rule is introduced in
Section~\ref{Sec:Algorithm}.

Formally, let \(I_k\) denote the length of batch \(k\). Set \(t_0=0\),
\(t_k=\sum_{\ell=1}^k I_\ell\), and
$
\mathcal T_k:=\{t_{k-1}+1,\ldots,t_k\}.
$ 
At the beginning of batch \(k\), sellers enter the batch with a focal price
vector \(\hat{\mathbf p}^k\). During the batch, seller \(i\) experiments locally
around her focal price according to
\[
p_i^t=\hat p_i^k+\delta_i^k Z_i^t,\qquad t\in\mathcal T_k,
\]
where \(\delta_i^k>0\) is the experimentation magnitude and
\(Z^t=(Z_1^t,\ldots,Z_n^t)\) is the vector of experimentation variables. The
variables \(Z_i^t\) may be binary, as in the standard time-randomized A/B or switchback
experiments in which each period is assigned to a treatment or control price;
multi-valued; or continuous. Importantly, as emphasized in the Introduction, experimentation variables may be statistically dependent across sellers, for example when sellers' repricing systems operate within shared implementation environments such as third-party repricers or common software infrastructure (e.g., see Example~\ref{Example:Two-Point}). 
We assume throughout the paper that there exist constants $B_i<\infty$ such that $|Z_i^t|\le B_i$ almost surely for all $i$ and $t$.

Our main result links the market feedback structure and the correlation structure of experimentation to the market's long-run outcome, which we characterize as a Conjectural Variations equilibrium. Hence, we now define this equilibrium concept.

\subsection{Conjectural Variations (CV) Equilibrium} \label{Sec:CV}

In this section, we define a conjectural variations (CV) equilibrium. Fix a conjecture matrix $A=(A_{ij})_{n\times n}$ with the convention that $A_{ii}=0$ for all $i\in[n]$, and let $A_{i\cdot}$ denote the $i$th row of $A$. Following the standard CV interpretation, $A_{ij}$ captures seller $i$'s conjecture about the local marginal response of rival $j$'s price to a small change in her own price. That is, for each $i$ and $j\neq i$, seller $i$ has a conjecture that an infinitesimal deviation $dp_i$ induces a contemporaneous co-movement $dp_j$ satisfying
\[
\frac{dp_j}{dp_i}=A_{ij},\qquad j\neq i.
\]
We refer to $A$ as the conjecture matrix, and define CV equilibrium given $A$ below.

Given $A$, consider a static pricing game without noise in which seller $i$ chooses a price $p_i\in[p_i^{l},p_i^{h}]$ to maximize revenue
$
r_i(\mathbf{p})=p_i\,\lambda_i(\mathbf{p}),
$
taking the conjectured local co-movement encoded by $A$ as given. Let $e_i$ denote the $i$th unit vector in $\mathbb{R}^n$. For an interior choice, seller $i$ considers a small increase in her own price $p_i$ and, taking the conjecture $A$ as given, evaluates the resulting marginal change in revenue when rivals' prices co-move according to $dp_j=A_{ij}\,dp_i$. This conjecture defines a direction in price space along which seller $i$ evaluates the marginal effect of changing $p_i$, namely
\[
v_i(A_{i\cdot})\;:=\;e_i\;+\;\sum_{j\neq i}A_{ij}\,e_j,
\]
which represents a unit increase in $p_i$ together with the conjectured contemporaneous adjustments of rivals' prices.
Therefore, under the first-order approach, an interior CV best reply is characterized by 
\begin{equation*} 
0= \nabla r_i(\mathbf{p})^\top v_i(A_{i\cdot}) = \partial_{p_i}r_i(\mathbf{p})+\sum_{j\ne i}A_{ij}\partial_{p_j}r_i(\mathbf{p})
=
\Big(\lambda_i(\mathbf{p})+p_i\,\partial_{p_i}\lambda_i(\mathbf{p})\Big)
+
\sum_{j\neq i} p_iA_{ij}\partial_{p_j}\lambda_i(\mathbf{p}). 
\end{equation*}
Hence an interior CV(\(A\)) best reply satisfies the first-order condition
\begin{equation}\label{eq:CV-FOC}
\lambda_i(\mathbf{p})\;+\;p_i\left(\partial_{p_i}\lambda_i(\mathbf{p})+\sum_{j\ne i}A_{ij}\partial_{p_j}\lambda_i(\mathbf{p})\right)=0.
\end{equation}

This leads to the following definition: 

\begin{definition}[CV$(A)$ equilibrium]
A CV$(A)$ equilibrium is a profile $\mathbf{p}\in\mathcal{P}$ such that, for each seller $i$, $p_i$ maximizes seller $i$'s revenue under the conjecture matrix $A$. In particular, given the sufficiency of the first-order approach (Assumption \ref{assu:cv-foc-sufficiency}), an interior equilibrium $\mathbf{p}$ is characterized by satisfying \eqref{eq:CV-FOC} for all $i\in[n]$.
\end{definition}

When $A=0$, the conjecture imposes no co-movement (i.e., $dp_j=0$ for all $j\neq i$), so seller $i$ behaves as if rivals' prices are locally fixed when she varies $p_i$. In this case, the CV$(0)$ notion reduces to the standard Nash equilibrium.  

CV equilibria are typically interpreted less as a literal description of adjustment over time and more as a mental mechanism that captures how sellers reason about their rivals’ reactions, because it describes ``dynamics'' in a static model \cite{tirole1988theory}. In this sense, sellers do not model the full strategic dynamics of the game but instead rely on local conjectures about how rivals’ actions co-move with their own. Our results show that CV equilibria can arise endogenously from the design of experimentation, even when sellers act independently and do not hold explicit local conjectures, thereby providing a learning foundation for the CV equilibria. 

We impose the following standard first-order sufficiency condition.

\begin{assumption}[First-order sufficiency for CV$(A)$]
\label{assu:cv-foc-sufficiency}
Fix a conjecture matrix $A=(A_{ij})_{n\times n}$ with $A_{ii}=0$. For each
seller $i$ and each price profile $\mathbf p\in\mathcal P$, if the $i$th
condition in \eqref{eq:CV-FOC} holds at $\mathbf p$, then $p_i$ is a CV$(A)$
best reply to $\mathbf p_{-i}$.
\end{assumption}

Assumption~\ref{assu:cv-foc-sufficiency} is the analogue, in the CV problem, of
the usual pseudo-concavity or single-peakedness condition imposed in pricing and
learning models to ensure that first-order conditions characterize optimal
prices. Here the relevant one-dimensional objective is seller $i$'s revenue
evaluated along the conjectured direction $v_i(A_{i\cdot})$. Thus the condition
rules out spurious stationary points along that direction. Lemma~\ref{lm:best_reply_satisfying_models}
in the Appendix verifies this condition for the linear and MNL demand
specifications studied below.

\begin{remark}
\label{rem:cv-foc-endogeneity}
One could instead impose a more primitive sufficient condition directly on the
conjectured path. Given a profile $\mathbf p$ and conjecture row $A_{i\cdot}$,
seller $i$ considers scalar deviations $s$ such that
$\mathbf p+s\,v_i(A_{i\cdot})\in\mathcal P$ and evaluates the one-dimensional
objective $(p_i+s)\lambda_i(\mathbf p+s\,v_i(A_{i\cdot}))$. At an interior
profile, the derivative of this objective with respect to $s$ at $s=0$ is
exactly the left-hand side of \eqref{eq:CV-FOC}. Thus, if every interior
stationary point of this one-dimensional problem is a global maximizer, then
Assumption~\ref{assu:cv-foc-sufficiency} follows. We state the assumption
directly in terms of the first-order condition \eqref{eq:CV-FOC}, because it is
cleaner and because this is the object identified by the learning dynamics. The
path-based formulation is equivalent in spirit, but it requires carrying the
feasible set of scalar deviations and the conjectured direction
 throughout the statement.
\end{remark}

\begin{remark}
\label{rem:boundary-cv}
For expositional simplicity, we state the CV first-order condition
\eqref{eq:CV-FOC} for interior solutions. If the relevant solution lies on the
boundary of $\mathcal P$, the corresponding one-sided KKT conditions replace
the interior condition. Under the same first-order sufficiency logic in
Assumption~\ref{assu:cv-foc-sufficiency}, these one-sided conditions
characterize a boundary CV$(A)$ best reply. Thus the convergence result provided in Theorem \ref{thm:main_theorem} extends
to boundary CV equilibria with the usual KKT interpretation.
\end{remark}

\subsection{Algorithm: Linear Demand Learning} \label{Sec:Algorithm}

In this section, we present our dynamic pricing algorithm: Linear Demand Learning (LDL). The algorithm follows a simple three-stage cycle: experimentation, estimation, and reoptimization. The experimentation stage is described in detail in Section \ref{sec:corr_feedback}. In the estimation stage, each seller fits a local linear demand model using only the most recent batch of data and the prices she observes, which depend on the feedback structure detailed in Section \ref{sec:corr_feedback}. In the reoptimization stage, the seller computes the myopic revenue-maximizing price implied by this fitted linear model and updates her focal price by moving partially toward this target. A parameter $u_i$ controls the amount of damping in the price update: a larger $u_i$ moves seller $i$ more aggressively toward the fitted myopic target, while a smaller $u_i$ produces more stable pricing dynamics.

Let \((\Omega,\mathcal F,\mathbb P)\) be the underlying probability space, and
let \(\mathcal F_t\) denote the full history generated by prices, demands, and
experimentation variables up to the end of period \(t\).\footnote{The filtration
\(\mathcal F_t\) is used to state measurability and
independence conditions. It is not seller \(i\)'s information set. Seller \(i\)'s
update uses only her own demand observations and the prices in
\(\mathcal R_i\).} The focal price \(\hat{\mathbf p}^k\) is
\(\mathcal F_{t_{k-1}}\)-measurable. The experimentation rule in batch \(k\)
may depend on \(\mathcal F_{t_{k-1}}\), but conditional on
\(\mathcal F_{t_{k-1}}\), the experimentation variables used are
independent of the contemporaneous demand shocks in that batch.

To keep realized prices feasible, define
$\mathcal P_i^k:=[p_i^l+B_i\delta_i^k,\;p_i^h-B_i\delta_i^k]$ and
$\mathcal P^k:=\prod_i\mathcal P_i^k$, and assume
$p_i^h-p_i^l>2B_i\delta_i^k$ for all $i,k$. The update of the focal price is
projected onto $\mathcal P_i^{k+1}$. All projections used in the paper are Euclidean projections onto the indicated interval or rectangle. 

For any index set $J\subseteq[n]$, we write
$\mathbf p_J^t:=(p_j^t:j\in J)$ for the corresponding price subvector, with
coordinates ordered increasingly by seller index. The same convention applies
to other vectors indexed by $J$, such as $\bar{\mathbf{p}}_J^k$. 
We describe the algorithm below.\footnote{We use the following sample-path conventions: if the
observable design matrix $\widehat G_i^k$ is singular, the OLS slope
coefficients are not uniquely defined, so the algorithm sets
$\hat z_i^{k+1}=\hat p_i^k$. If the fitted own-price slope satisfies
$\hat\beta_i^{k+1}\le\underline\beta_i$ for a small tolerance
$\underline\beta_i\ge0$ the target $\hat z_i^{k+1}$ is instead replaced by the default feasible
target $p_i^h$ before projection. These conventions only make the algorithm
well defined on every sample path. Under the assumptions of
Theorem~\ref{thm:main_theorem}, they occur only on the bad event for all
sufficiently large batches and do not affect the limit.}

\begin{algorithm}
\caption{Linear Demand Learning with Partial Price Feedback (LDL)}
\label{alg:partial-feedback-ldl}
\begin{algorithmic}
\REQUIRE $\hat p_i^1\in\mathcal P_i^1$, feedback sets
$\{\mathcal S_i\}_{i=1}^n$, learning rates $u_i\in(0,1]$, batch sizes
$\{I_k\}$, experimentation magnitudes $\{\delta_i^k\}$, experimentation
variables $\{Z^t\}_{t\ge1}$, and slope tolerances $\underline\beta_i\ge0$
\STATE $t_0\gets0$
\FOR{$k=1,2,\ldots$}
\STATE Set $t_k\gets t_{k-1}+I_k$ and
$\mathcal T_k\gets\{t_{k-1}+1,\ldots,t_k\}$
\STATE \textbf{Step 1: Experimentation}
\FOR{$t\in\mathcal T_k$}
\STATE Generate the experimentation vector $Z^t=(Z_1^t,\ldots,Z_n^t)$
\STATE Each seller $i$ posts $p_i^t\gets\hat p_i^k+\delta_i^kZ_i^t$
\STATE Each seller $i$ observes $D_i^t$ and the prices
$\{p_r^t:r\in\mathcal R_i\}$
\ENDFOR
\STATE \textbf{Step 2: Estimation}
\FOR{each seller $i$}
\STATE Set $\bar{\mathbf p}_{\mathcal R_i}^k\gets I_k^{-1}\sum_{t\in\mathcal T_k}
\mathbf p_{\mathcal R_i}^t$
\STATE Set
\[
\widehat G_i^k
\gets
I_k^{-1}\sum_{t\in\mathcal T_k}
(\mathbf p_{\mathcal R_i}^t-\bar{\mathbf p}_{\mathcal R_i}^k)
(\mathbf p_{\mathcal R_i}^t-\bar{\mathbf p}_{\mathcal R_i}^k)^\top .
\]
\IF{$\widehat G_i^k$ is singular}
\STATE Set $\hat z_i^{k+1}\gets\hat p_i^k$
\ELSE
\STATE Compute the unique OLS coefficients
\[
(\hat\alpha_i^{k+1},\{\hat b_{ir}^{k+1}:r\in\mathcal R_i\})
\in
\arg\min_{\alpha,\{b_r\}}
\sum_{t\in\mathcal T_k}
\left[D_i^t-\alpha-\sum_{r\in\mathcal R_i}b_rp_r^t\right]^2 .
\]
\STATE Set $\hat\beta_i^{k+1}\gets-\hat b_{ii}^{k+1}$ and
$\hat\theta_{ij}^{k+1}\gets\hat b_{ij}^{k+1}$ for $j\in\mathcal S_i$
\STATE \textbf{Step 3: Reoptimization}
\STATE Set
\[
\hat z_i^{k+1}
\gets
\begin{cases}
p_i^h, & \hat\beta_i^{k+1}\le\underline\beta_i,\\[1ex]
\dfrac{
\hat\alpha_i^{k+1}
+
\sum_{j\in\mathcal S_i}\hat\theta_{ij}^{k+1}\bar p_j^k
}{
2\hat\beta_i^{k+1}
}, & \hat\beta_i^{k+1}>\underline\beta_i .
\end{cases}
\]
\ENDIF
\STATE Update
\[
\hat p_i^{k+1}
\gets
\mathrm{proj}_{\mathcal P_i^{k+1}}
\left((1-u_i)\hat p_i^k+u_i\hat z_i^{k+1}\right).
\]
\ENDFOR
\ENDFOR
\end{algorithmic}
\end{algorithm}

\subsection{Main Theorem: Convergence}\label{sec:convergence}

We first describe the population objects that determine the limiting equilibrium. For each batch $k$, define the centered experimentation variables $\tilde Z_i^{t,k}:=Z_i^t-I_k^{-1}\sum_{s\in\mathcal T_k}Z_i^s$. Note that because the regressions include intercepts, mean experimentation levels do not induce bias; only centered within-batch co-movement matters.

For seller $i$, the empirical covariance matrix of the observed experimentation is $I_k^{-1}\sum_{t\in\mathcal T_k} \tilde Z_{\mathcal R_i}^{t,k} (\tilde Z_{\mathcal R_i}^{t,k})^\top$, which captures the variation in the prices included in her regression. For an unobserved rival $\ell\in\mathcal U_i$, the empirical
covariance vector
$
I_k^{-1}\sum_{t\in\mathcal T_k}
\tilde Z_{\mathcal R_i}^{t,k}\tilde Z_\ell^{t,k}
$ 
captures how seller $\ell$'s unobserved experimentation co-moves with the
experimentation variables observed by seller $i$. When these objects converge (see Assumption \ref{assu:A-star}),
write their deterministic limits as $\Sigma_i^\star$ and $\mathbf c_{i\ell}^\star$.
Thus $\Sigma_i^\star$ is the limiting observed-design covariance matrix for
seller $i$, and $\mathbf c_{i\ell}^\star$ is the limiting covariance vector between
unobserved rival $\ell$ and seller $i$'s observed price regressors.
Given these limits, define
$\pi_{i\ell}^\star:=(\Sigma_i^\star)^{-1}\mathbf c_{i\ell}^\star$ for
$\ell\in\mathcal U_i$. This is a population partial linear-projection
coefficient, with coordinates indexed by \(\mathcal R_i\): for
\(r\in\mathcal R_i\), \([\pi_{i\ell}^\star]_r\) denotes the coordinate
associated with seller \(r\). Thus \([\pi_{i\ell}^\star]_r\) measures how much seller
$\ell$'s unobserved experimentation is linearly predicted by seller $r$'s
observed experimentation, after controlling for the other observed price
movements in $\mathcal R_i$.

With this notation, define the conjecture matrix \(A^{\mathcal G,\star}\), the
statistical object that links experimentation and feedback to CV equilibria, by
\begin{equation}\label{eq:A-star-partial-feedback}
A_{ij}^{\mathcal G,\star}
:=
\begin{cases}
[\pi_{ij}^\star]_i, & j\in\mathcal U_i,\\
0, & j\in\mathcal R_i.
\end{cases}
\end{equation}
If rival $j$'s price is observed by seller $i$, it is included directly as a regressor in seller $i$'s linear demand model. Because this observed price is held fixed when seller $i$ computes the fitted myopic target, the corresponding conjecture coefficient is zero. Conversely, if rival $j$ is unobserved, $j$'s price acts as an omitted variable in the OLS regression. The coefficient $[\pi_{ij}^\star]_i$ captures exactly the portion of $j$'s unobserved price experimentation that is linearly predicted by seller $i$'s own experimentation, after controlling for the other observed prices. Hence, the limiting conjecture matrix is endogenously determined by the interaction between the feedback structure $\mathcal G$ and the covariance structure of experimentation. As we establish in Theorem~\ref{thm:main_theorem}, under the assumptions we now present, the LDL algorithm converges (in expectation) to this CV$(A^{\mathcal G,\star})$ equilibrium.

For any conjecture matrix $A=(A_{ij})_{n\times n}$ with $A_{ii}=0$, define
\begin{equation}\label{eq:beta_alpha}
\beta_i^{(A)}(\mathbf p)
:=
-\left(\partial_{p_i}\lambda_i(\mathbf p)
+
\sum_{j\ne i}A_{ij}\partial_{p_j}\lambda_i(\mathbf p)\right)
\end{equation}
which is the adjusted own-price demand slope faced by seller \(i\)
under conjecture \(A\).

\begin{assumption}
\label{assu:A-star}
For every seller
$i$, $I_k^{-1}\sum_{t\in\mathcal T_k}\tilde Z_{\mathcal R_i}^{t,k}
(\tilde Z_{\mathcal R_i}^{t,k})^\top\xrightarrow{a.s.}\Sigma_i^\star$ as
$k\to\infty$. In addition, for every seller $i$ and every unobserved rival
$\ell\in\mathcal U_i$,
$I_k^{-1}\sum_{t\in\mathcal T_k}\tilde Z_{\mathcal R_i}^{t,k}
\tilde Z_\ell^{t,k}\xrightarrow{a.s.}\mathbf c_{i\ell}^\star$ as $k\to\infty$.
The matrix $\Sigma_i^\star$ is positive definite for every seller $i$.
Finally, for \(A^{\mathcal G,\star}\) defined in
\eqref{eq:A-star-partial-feedback}, there exists \(\beta_{\min}>0\) such that
$
\beta_i^{(A^{\mathcal G,\star})}(\mathbf p)\ge \beta_{\min}
$ 
for every seller \(i\) and every \(\mathbf p\in\mathcal P\), and the slope
tolerances in Algorithm~\ref{alg:partial-feedback-ldl} satisfy
\(\underline\beta_i<\beta_{\min}/2\) for all \(i\).
\end{assumption}

Assumption~\ref{assu:A-star} is mainly a technical regularity condition for the
LDL dynamics. The stabilization
condition is a law-of-large-numbers requirement for the empirical design: within
large batches, the second moments that determine the OLS projection converge to
deterministic limits. It is satisfied, for example, by bounded i.i.d.\
experimentation within each batch, conditional on the past, and by many
stationary designs with stable second moments.
Positive definiteness of $\Sigma_i^\star$ rules out asymptotic collinearity
among the observed price regressors. The final slope condition ensures that the
limiting own-price slope learned by OLS is uniformly positive. 

When $\beta_i^{(A)}(\mathbf p)>0$, let
\[
z_i^{(A)}(\mathbf p)
:=
\frac{\lambda_i(\mathbf p)+p_i\beta_i^{(A)}(\mathbf p)}
{2\beta_i^{(A)}(\mathbf p)} 
\]
be the myopic revenue-maximizing price under the corresponding linear demand
approximation and conjecture $A$. As usual, for a differentiable vector-valued function $z$, write $Dz(\mathbf p)$
for its Jacobian at $\mathbf p$.\footnote{For
$z:\mathcal P\to\mathbb R^n$, $Dz(\mathbf p)$ is the matrix
$(\partial z_i(\mathbf p)/\partial p_j)_{i,j\in[n]}$. For a matrix
$X=(x_{ij})$, $\|X\|_\infty:=\max_i\sum_j|x_{ij}|$ is the maximum row-sum
norm.}

Some stability condition is needed to obtain convergence in multi-agent learning
dynamics, especially given the misspecification in our model.
Assumption~\ref{assu:ldl-contraction} below is the key stability condition that
allows us to prove convergence despite the sellers' misspecification from
ignoring some competitive effects and fitting a local linear demand model. As we discuss in Section~\ref{sec:examples}, this stability condition has a
simple economic interpretation: own-price demand effects must be strong enough
relative to aggregate cross-price effects and local nonlinearities of demand. In
linear demand models, where there is no demand-curvature misspecification, this
becomes the usual diagonal-dominance condition. In nonlinear models, it requires
demand to be sufficiently close to locally linear in the relevant price region.
Section~\ref{sec:examples} verifies this condition for linear demand and gives
simple sufficient conditions for MNL demand over a wide range of parameters.

In addition, convergence is not the only sense in which our algorithm is
informative. In Section~\ref{sec:beyond-contraction}, we show that even without
Assumption~\ref{assu:ldl-contraction}, the conjecture matrix induced by the
experimentation and feedback structure still determines the  direction
of the price updates.

\begin{assumption}
\label{assu:ldl-contraction}
For $A^{\mathcal G,\star}$ defined in
\eqref{eq:A-star-partial-feedback},
\[
\sup_{\mathbf p\in\mathcal P}
\|Dz^{(A^{\mathcal G,\star})}(\mathbf p)\|_\infty<1 .
\]
\end{assumption}

Under Assumption~\ref{assu:ldl-contraction}, let
\(\mathbf p^{(A^{\mathcal G,\star})}\) denote the unique fixed point in
\(\mathcal P\) of the deterministic population LDL map
$
\mathrm{proj}_{\mathcal P}\!\left((I-U)\mathbf p
+Uz^{(A^{\mathcal G,\star})}(\mathbf p)\right)$ where
$U:=\operatorname{diag}(u_1,\ldots,u_n)$. 

We now present our main theorem: 

\begin{theorem}
\label{thm:main_theorem}
Consider the matrix $A^{\mathcal G,\star}$ defined in
\eqref{eq:A-star-partial-feedback}. 
Suppose that Assumption~\ref{assu:cv-foc-sufficiency} holds for
$A=A^{\mathcal G,\star}$ and that Assumptions~\ref{assu:A-star}
and~\ref{assu:ldl-contraction} hold. For the parameters in
Algorithm~\ref{alg:partial-feedback-ldl}, suppose
$\delta_i^k\downarrow0$, $\delta_i^k/\delta_j^k$ is uniformly bounded and
converges to $1$ for all $i,j\in[n]$,
$\delta_i^k I_k^{1/2}/\log(eI_k)\to\infty$ for all $i\in[n]$, and
$u_i\in(0,1]$ for all $i\in[n]$.

Then the focal prices generated when all sellers use
Algorithm~\ref{alg:partial-feedback-ldl} converge in expectation:
\[
\mathbb E\left[
\left\|\hat{\mathbf p}^k-\mathbf p^{(A^{\mathcal G,\star})}\right\|_\infty
\right]\to0,
\qquad k\to\infty,
\]
 Moreover, if $\mathbf p^{(A^{\mathcal G,\star})}$ lies in
$\mathrm{int}(\mathcal P)$, then it is a CV$(A^{\mathcal G,\star})$
equilibrium.\footnote{Boundary solutions are interpreted as
in Remark~\ref{rem:boundary-cv}.}
\end{theorem}

We now discuss two benchmark cases of Theorem~\ref{thm:main_theorem}.

\textbf{Full price feedback.} If every seller observes every rival's price, then
$\mathcal U_i=\emptyset$ for every seller $i$, so $A^{\mathcal G,\star}=0$:
all rival price movements are controlled for directly. Hence LDL converges to
the Nash equilibrium even under correlated experimentation. Full feedback is generally a strong assumption in competitive markets with many
sellers, since it requires including every other
seller's price in the learning procedure. Hence, even if only some rival prices
are omitted, correlated experimentation among those omitted rivals can still
generate a nonzero conjecture matrix and lead the dynamics to a CV$(A)$
equilibrium rather than the Nash equilibrium.

\textbf{Bandit feedback.} At the other extreme, if $\mathcal S_i=\emptyset$ for
every seller $i$, then seller $i$ controls only for her own price. In this case,
for $j\ne i$,
\[
A_{ij}^{\mathcal G,\star}
=
\frac{\operatorname{Cov}^\star(Z_i,Z_j)}
{\operatorname{Var}^\star(Z_i)}
:=
\frac{\lim_k I_k^{-1}\sum_{t\in\mathcal T_k}\tilde Z_i^{t,k}\tilde Z_j^{t,k}}
{\lim_k I_k^{-1}\sum_{t\in\mathcal T_k}(\tilde Z_i^{t,k})^2}.
\]
Thus the conjecture is the limiting regression coefficient of rival $j$'s
experimentation on seller $i$'s experimentation, which is exactly the
omitted-variable-bias term: when seller $j$'s unobserved experimentation
co-moves with seller $i$'s experimentation, part of seller $j$'s cross-price
effect is attributed to seller $i$'s own-price effect, thereby connecting biased
demand learning to the CV equilibrium that is learned.

We next illustrate how familiar switchback experimentation designs translate into the conjecture coefficients.

\begin{example}[Two-point switchback experiments] \label{Example:Two-Point}
The standard two-point A/B design is recovered by taking
$Z_i^t\in\{0,1\}$, with $\mathbb P(Z_i^t=1)=q_i \in (0,1)$, where $Z_i^t=1$ means that
seller $i$ assigns period $t$ to the treatment price. In the bandit-feedback
case,
\[
A_{ij}^{\mathcal G,\star}
=
\frac{\mathrm{Cov}(Z_i,Z_j)}{\mathrm{Var}(Z_i)}
=
\mathbb P(Z_j=1\mid Z_i=1)
-
\mathbb P(Z_j=1\mid Z_i=0),
\]
where the probabilities are under the limiting within-batch experimentation
law given in Assumption \ref{assu:A-star}. Thus the coefficient is the difference in the probability that seller $j$
assigns a period to treatment when seller $i$ assigns that period to treatment
versus when seller $i$ does not.

Many reduced-form models can generate such correlation. One stylized way to
capture the common repricing windows discussed in the Introduction is to suppose
that some sellers share the same implementation environment (such as a
third-party repricer or shared software infrastructure). Let $g(i)$ denote seller $i$'s implementation group. For each group $g$, draw a common implementation state
$B_g^t\sim\mathrm{Bernoulli}(\pi_g)$, where $B_g^t=1$ represents a period in which the shared implementation environment increases the probability that sellers in group $g$ post their treatment price.\footnote{A more literal model, which could also
capture API scheduling, would specify a continuous-time process in which
repricer batch jobs, API throttling, or other software constraints create common
update opportunities. We do not model this explicitly; the common state
\(B_g^t\) is a reduced-form representation of the resulting period-level
co-movement in experimentation decisions.}  Conditional
on the group states, experimentation decisions are independent across sellers,
with
\[
\mathbb P(Z_i^t=1\mid B_{g(i)}^t=1)=q_i+(1-\pi_{g(i)})d_i,\qquad
\mathbb P(Z_i^t=1\mid B_{g(i)}^t=0)=q_i-\pi_{g(i)}d_i,
\]
where $d_i$ measures how strongly seller $i$'s period-level treatment
probability responds to the common implementation state, and is chosen so that
both conditional probabilities lie in $[0,1]$. Seller $i$'s marginal
experimentation probability remains $q_i$, but the common implementation state
makes treatment assignments co-move within a group. If sellers $i$ and $j$ are
in the same group $g$, then\footnote{Indeed, note that the construction implies
$\mathbb E[Z_i^t\mid B_g^t]=q_i+d_i(B_g^t-\pi_g)$ and
$\mathbb E[Z_j^t\mid B_g^t]=q_j+d_j(B_g^t-\pi_g)$. Conditional on $B_g^t$,
the experimentation decisions are independent, so the law of total covariance
gives
$\mathrm{Cov}(Z_i^t,Z_j^t)=
\mathrm{Cov}(\mathbb E[Z_i^t\mid B_g^t],
\mathbb E[Z_j^t\mid B_g^t])
=\pi_g(1-\pi_g)d_id_j$.
Since $Z_i^t$ is Bernoulli with marginal probability $q_i$,
$\mathrm{Var}(Z_i^t)=q_i(1-q_i)$, yielding the stated expression.}
\[
A_{ij}^{\mathcal G,\star}
=
\frac{\pi_g(1-\pi_g)d_id_j}{q_i(1-q_i)}.
\]
If their implementation states are independent, then
$A_{ij}^{\mathcal G,\star}=0$.
\end{example}

\begin{example}[Three-point switchback experiments]
A three-point switchback design, which can be natural in strategic environments, uses  $Z_i^t\in\{-1,0,1\}$ corresponding to
$\hat p_i^k-\delta_i^k$, $\hat p_i^k$, and
$\hat p_i^k+\delta_i^k$; see \citet{wu2024switchback}. In the bandit-feedback
case, suppose seller $i$'s design is symmetric:
$\mathbb P(Z_i=1)=\mathbb P(Z_i=-1)=q$, where \(0<q\le 1/2\). Then
\[
A_{ij}^{\mathcal G,\star}
=
\frac12\left(
\mathbb E[Z_j\mid Z_i=1]
-
\mathbb E[Z_j\mid Z_i=-1]
\right).
\]
Thus the induced conjecture is the conditional swing in seller $j$'s
experimentation as seller $i$ moves from the low experimental price to the high
experimental price, normalized by the length of that two-step movement.\footnote{
Under symmetry, $\mathbb E[Z_i]=0$ and $\mathrm{Var}(Z_i)=2q$. Hence
$\mathrm{Cov}(Z_i,Z_j)=
q\,\mathbb E[Z_j\mid Z_i=1]
-
q\,\mathbb E[Z_j\mid Z_i=-1]$, and dividing by $2q$ gives the expression in the
text.} Independent switchbacks give $A_{ij}^{\mathcal G,\star}=0$, while a
common switchback clock with $Z_j^t=Z_i^t$ gives
$A_{ij}^{\mathcal G,\star}=1$.
\end{example}

\subsection{Convergence to Nash Equilibrium}

A particularly important case is when the induced conjecture matrix is zero.
This happens whenever seller \(i\)'s observed experimentation variables carry no
limiting linear information about the experimentation variables of rivals whose
prices seller \(i\) does not observe. Beyond the full-feedback case discussed
above, independent experimentation across sellers is the simplest example: then
the unobserved experimentation movements are asymptotically orthogonal to the
observed ones, so the omitted-variable bias in the learned coefficients
disappears. The following corollary formalizes this zero-bias case.

\begin{corollary}
\label{corr:nash}
Suppose that all the assumptions of
Theorem~\ref{thm:main_theorem} hold. Suppose, in addition, that
for every seller $i$ and every unobserved rival $\ell\in\mathcal U_i$, the
limiting covariance vector satisfies $\mathbf c_{i\ell}^\star=0$. Then
$A^{\mathcal G,\star}=0$. Consequently, the sequence
$\{\hat{\mathbf p}^k:k\ge1\}$ generated when all sellers use
Algorithm~\ref{alg:partial-feedback-ldl} satisfies
\[
\mathbb E\big[\|\hat{\mathbf p}^k-\mathbf p^{(0)}\|_\infty\big]\to0,
\qquad k\to\infty.
\]
If $\mathbf p^{(0)}$ lies in $\mathrm{int}(\mathcal P)$, then the limiting CV$(0)$ equilibrium
coincides with the Nash equilibrium.
\end{corollary}

\section{Proof Intuition and Beyond Linear Demand Learning}
\label{sec:intuition-beyond-ldl}

Theorem \ref{thm:main_theorem} is stated for the LDL algorithm, in which each seller fits a
linear demand model by OLS. This linear regression step is useful because it
produces a simple closed-form conjecture matrix: the entries of
\(A^{\mathcal G,\star}\) are partial linear-projection coefficients of
unobserved experimentation on the price variation observed by a seller.
However, the mechanism behind the CV limit and the connection to the feedback structure and experimentation design is not an artifact of linear models. The key object is the demand relationship that a seller can
learn from the data she actually observes.
For seller \(i\), define the batch-\(k\) conditional mean demand by
\begin{equation}
\label{eq:conditional-mean-demand}
m_{i,k}(\mathbf x_{\mathcal R_i};\hat{\mathbf p}^k)
:=
\mathbb E\!\left[
D_i^t
\,\middle|\,
\mathbf p_{\mathcal R_i}^t=\mathbf x_{\mathcal R_i},\mathcal F_{t_{k-1}}
\right].
\end{equation}
This is the demand object generated by seller \(i\)'s data in batch \(k\). Since
the demand shock has conditional mean zero,
$
m_{i,k}(\mathbf x_{\mathcal R_i};\hat{\mathbf p}^k)
=
\mathbb E\!\left[
\lambda_i(\mathbf x_{\mathcal R_i},\mathbf p_{\mathcal U_i}^t)
\,\middle|\,
\mathbf p_{\mathcal R_i}^t=\mathbf x_{\mathcal R_i},\mathcal F_{t_{k-1}}
\right]
$. 
Thus, when seller \(i\) does not observe all rival prices, the object she can
learn is not the primitive demand function \(\lambda_i(\mathbf p)\). It is the
mean demand after averaging over the unobserved prices according to their
conditional distribution given the observed prices. Under full price feedback,
there are no unobserved prices and this conditional mean coincides with the
true demand function evaluated at the observed price vector. Under partial
feedback, however, unobserved rival prices enter through their statistical
co-movement with the prices seller \(i\) observes.

To see the connection to CV, consider for simplicity a common experimentation
magnitude \(\delta^k\). Under standard regularity assumptions, at an experimental value
\(\mathbf x_{\mathcal R_i}=\hat{\mathbf p}_{\mathcal R_i}^k+\delta^k \mathbf z_{\mathcal R_i}\), the
regular conditional mean demand can be written as
\[
m_{i,k}(\hat{\mathbf p}_{\mathcal R_i}^k+\delta^k z_{\mathcal R_i};
\hat{\mathbf p}^k)
=
\mathbb E\!\left[
\lambda_i(
\hat{\mathbf p}_{\mathcal R_i}^k+\delta^k \mathbf z_{\mathcal R_i},
\hat{\mathbf p}_{\mathcal U_i}^k+\delta^k Z_{\mathcal U_i}^t)
\,\middle|\,
Z_{\mathcal R_i}^t= \mathbf z_{\mathcal R_i},
\mathcal F_{t_{k-1}}
\right]
\]
and the first-order expansion around the focal price gives
\begin{align}
m_{i,k}(\hat{\mathbf p}_{\mathcal R_i}^k+\delta^k z_{\mathcal R_i};
\hat{\mathbf p}^k)
&=
\lambda_i(\hat{\mathbf p}^k)
+
\delta^k
\sum_{r\in\mathcal R_i}
\partial_{p_r}\lambda_i(\hat{\mathbf p}^k)z_r
\nonumber\\
&\quad+
\delta^k
\sum_{\ell\in\mathcal U_i}
\partial_{p_\ell}\lambda_i(\hat{\mathbf p}^k)
\mathbb E[
Z_\ell^t
\mid
Z_{\mathcal R_i}^t= \mathbf z_{\mathcal R_i},
\mathcal F_{t_{k-1}}]
+
O((\delta^k)^2).
\label{eq:conditional-mean-local-expansion}
\end{align}
The first term is the local demand level at the focal price. The second term is
the direct effect of the price coordinates included in seller \(i\)'s learning
problem. The third term is the effect of unobserved rivals' price movements,
averaged according to how those unobserved movements co-move with the observed
experimental prices. This third term is the source of the
conjectural-variation component. Even without imposing a linear learning rule,
a seller who learns a local marginal effect from this conditional demand object
does not learn only the primitive own-price derivative. She also incorporates
the demand effect of unobserved rivals' price movements that predictably
accompany the observed price variation. Thus the learned marginal effect behaves
as if seller \(i\)'s own price movement were accompanied by induced movements
in unobserved rivals' prices. This is precisely the idea of a conjectural variation.

The OLS demand learning identifies exactly these conjectures. In the stabilized
population design of Assumption~\ref{assu:A-star}, the own-price coefficient
learned by OLS is, to first order,
$ 
\partial_{p_i}\lambda_i(\hat{\mathbf p}^k)
+
\sum_{j\ne i}
A_{ij}^{\mathcal G,\star}
\partial_{p_j}\lambda_i(\hat{\mathbf p}^k),
$ 
where \(A_{ir}^{\mathcal G,\star}=0\) for observed rivals
\(r\in\mathcal S_i\), and
\(A_{i\ell}^{\mathcal G,\star}=[\pi_{i\ell}^{\star}]_i\) for unobserved rivals
\(\ell\in\mathcal U_i\). Thus, OLS turns the linear projection of unobserved  experimentation on observed
experimentation into the conjecture matrix
\(A^{\mathcal G,\star}\) defined in \eqref{eq:A-star-partial-feedback}.
Observed rivals do not generate conjectural-variation coefficients because their
prices are included in the regression and are held fixed when seller \(i\)
computes the fitted myopic target.

Using this learned slope to differentiate the fitted local revenue (price times learned local demand) yields the local marginal revenue
\[
\lambda_i(\hat{\mathbf p}^k)
+
\hat p_i^k
\left(
\partial_{p_i}\lambda_i(\hat{\mathbf p}^k)
+
\sum_{j\ne i}
A_{ij}^{\mathcal G,\star}
\partial_{p_j}\lambda_i(\hat{\mathbf p}^k)
\right)
+
o(1).
\]
This is the CV marginal-revenue expression associated with the conjecture
matrix \(A^{\mathcal G,\star}\). Thus the algorithm behaves as if seller \(i\)
believes that an incremental change in her own price is accompanied by
contemporaneous movements of unobserved rivals' prices according to the
coefficients in \(A^{\mathcal G,\star}\) that are induced by the interaction between the feedback
structure and the correlation structure of experimentation.
The rest of the convergence argument formalizes this intuition. 

This also clarifies the role of nonlinear demand. LDL fits a linear model, but
only as a local approximation inside each shrinking experimental batch. The
leading first-order terms are governed by the derivatives of the true demand
function and by the conditional co-movement of observed and unobserved
experimentation. Nonlinearities enter through the local approximation error,
which is controlled by shrinking the experimentation magnitude. Thus the CV
limit is not driven by a globally linear demand assumption. Linear regression  provides a tractable and explicit way to identify the induced
conjecture matrix \(A^{\mathcal G,\star}\) from the observed experimental
variation.

\section{Discussion on Assumptions and Implications}
\label{sec:supra_competitive}

In this section we discuss some of our modeling assumptions and implications of our results.

\textbf{Exogeneity of feedback and experimentation.}
We assume that the feedback structure and the correlation structure of
experimentation are exogenous and focus on the long-run prices generated given
these structures. One could, in principle, study a meta-game in which sellers
choose whom to monitor or how to experiment in order to influence the induced
equilibrium. We do not model such a game here. Doing so would require sellers
to have strong information about demand primitives, rivals' experimentation
rules, and the feedback structures used by other sellers. These are objects
that are generally not known in the learning environment we study. Moreover, in
the Appendix, we show that even in simple two-seller linear-demand environments
there is no dominant choice of feedback or experimentation intensity: observing another seller or experimenting more frequently can
increase or decrease revenues depending on the market primitives and on the
other seller's behavior. Intuitively, these choices affect revenues through the induced conjecture matrix, and even when larger conjectures raise equilibrium prices, an individual seller's revenue need not be monotone because higher prices can reduce demand  too much. Thus, we believe that it would be quite hard in
practice for a seller to be strategic about these choices. 

In addition, in marketplace settings, if a platform or repricer deliberately used detailed
cross-seller information to influence market prices through feedback or
experimentation design, this would raise separate legal and regulatory
questions. We therefore take these structures as given and study their
implications for equilibrium selection.

\textbf{Equilibrium implications.}
Our main result (Theorem~\ref{thm:main_theorem}) shows that the correlation
structure of experimentation and the feedback structure jointly determine the
conjecture matrix \(A^{\mathcal G,\star}\), and therefore the long-run
equilibrium prices. 

Proposition~\ref{Prop:Price_increase} provides a useful
comparative-static benchmark. Under natural strategic-complementarity conditions
stated there, the extremal CV equilibria are coordinatewise nondecreasing in
the conjecture matrix \(A\). Hence, when the CV equilibrium is unique, larger
nonnegative conjectures lead to higher equilibrium prices. 
This result is especially transparent under bandit feedback. If
\(\mathcal S_i=\emptyset\) for every seller, then
\(A_{ij}^{\mathcal G,\star}
=\operatorname{Cov}^\star(Z_i,Z_j)/\operatorname{Var}^\star(Z_i)\).
Thus, under the conditions of
Proposition~\ref{Prop:Price_increase}, positively correlated experimentation
raises equilibrium prices relative to the Nash benchmark \(A=0\) that is generated by the full feedback case. 

With partial feedback, however, positive experimentation correlation does not
necessarily imply prices above Nash. The entries of \(A^{\mathcal G,\star}\)
are partial projection coefficients, not pairwise correlations. Hence, after
conditioning on the rival prices a seller observes, an omitted rival's
experimentation can be negatively associated with the seller's own
experimentation even when all pairwise correlations are positive. 
Consequently, under partial feedback, positively correlated experimentation can
generate negative conjectures and may lower equilibrium prices relative to the
Nash benchmark. Hence, as feedback expands from bandit feedback toward full
feedback, equilibrium prices need not move monotonically toward Nash: they can
decrease below the Nash level under an intermediate feedback structure and then
increase back toward Nash as feedback becomes full. 

This situation can arise in a market with one technologically sophisticated
major seller. In the Appendix, we consider a feedback structure in which smaller
sellers monitor the major seller but not each other, while the major seller has
the technology to track all sellers. We note that the assumption that the major seller tracks all smaller sellers is not essential
for the negative conjectures among the smaller sellers: if the major seller did
not track them, the conjectures among smaller sellers would remain negative,
while the major seller's own conjectures would be positive in this example. In this case, all experimentation
variables can be positively correlated unconditionally, but after controlling
for the major seller's price, the residual co-movement among smaller sellers is
negative. This illustrates that the sign of the induced conjectures is determined by the
resulting residual correlations, not by the feedback structure alone. The example generates only negative nonzero conjectures; hence, under
the monotone comparative-statics conditions of
Proposition~\ref{Prop:Price_increase}, the induced CV equilibrium prices are
 below the Nash benchmark.

We now present the formal details needed to state Proposition \ref{Prop:Price_increase}. 

Let
$
\mathcal A\subseteq
\{A\in\mathbb R^{n\times n}:A_{ii}=0,\ \forall i\in[n]\}
$ 
be a nonempty compact set of admissible conjecture matrices 
endowed with the coordinatewise partial order $\preceq$:
$
A\preceq A'$ if and only if $ A_{ij}\le A'_{ij}\ \ \forall\,i,j\in[n].
$

Recall that \(\mathcal P:=\prod_{i=1}^n[p_i^l,p_i^h]\subset\mathbb R_+^n\)
is endowed with the coordinatewise product order \(\le\). 
 Let $G:\mathcal P\times\mathcal A\to\mathbb R^n$ denote the marginal-revenue (FOC) mapping, i.e., 
\[
G_i(\mathbf{p};A)
:=
\lambda_i(\mathbf{p})
+
p_i\Big(\partial_{p_i}\lambda_i(\mathbf{p})+\sum_{j\neq i}A_{ij}\,\partial_{p_j}\lambda_i(\mathbf{p})\Big).
\]
Note that $G_{i}$ is continuous for every $i$.

 \begin{proposition} \label{Prop:Price_increase}
Assume that Assumption~\ref{assu:cv-foc-sufficiency} holds for every
\(A\in\mathcal A\).     Assume that for each $i$, each $j\neq i$, and for all $(\mathbf{p},A)\in\mathcal P\times\mathcal A$ we have 
$\partial_{p_j}G_i(\mathbf{p};A)\ \ge\ 0$. In addition, assume that for each $i$, all $\mathbf{p}_{-i}\in\prod_{j\neq i}[p_j^l,p_j^h]$,
and all $A\in\mathcal A$,
 we have\footnote{This condition rules out boundary fixed points so fixed points coincide
with interior solutions to $G(\mathbf{p};A)=0$. This condition can be established in linear and MNL demand models under a relevant set of conjectures.}
$G_i( p_i^l,\mathbf{p}_{-i};A)>0$
 and 
$G_i( p_i^h,\mathbf{p}_{-i};A)<0$. 

Then the lowest and highest (interior) solutions to $G(\mathbf{p};A)=0$, i.e., the extremal CV equilibria, are coordinatewise nondecreasing in $A$. If the CV equilibrium is unique, then $\mathbf{p}(A)$ is coordinatewise nondecreasing in $A$.
 \end{proposition}

\section{Convergence Rate} \label{sec:converge_rate}

While Theorem~\ref{thm:main_theorem} establishes convergence,
it does not quantify the speed of convergence. We now state a finite-sample
rate for fixed $n$. The rate is governed by the usual bias-variance tradeoff:
the local linearization error is of order $\delta^k$, whereas the statistical
error in the fitted slope is of order
$\sqrt{\log(eI_k)}/(\delta^k\sqrt{I_k})$. Balancing these terms gives
$\delta^k=(\log(eI_k)/I_k)^{1/4}$, which leads to a root-mean-squared price
error of order $\widetilde O(T^{-1/4})$ and hence a mean squared price error of
order $\widetilde O(T^{-1/2})$.

The result applies to the bounded local experimentation and partial-feedback
structures covered by Theorem~\ref{thm:main_theorem}. Thus the
same rate describes convergence to the CV equilibrium selected by correlated
experimentation. It also covers the Nash case, which arises when the induced
conjecture matrix is zero.\footnote{This rate matches the state-of-the-art mean
squared error rate for broader classes of bandit-feedback games, e.g.,
\citet{ba2025doubly}. Their algorithms use mirror descent with randomized
gradient estimators, whereas here the rate is obtained by price experimentation
and repeated least-squares fitting of local linear demand models.}

We now introduce some notation needed to present the next theorem. 
Let $\mathbf p(T)$ denote the vector of prices posted in period $T$ when all
sellers use Algorithm~\ref{alg:partial-feedback-ldl}. As usual,
$f(T)=\widetilde O(g(T))$ means
$f(T)\le Cg(T)\log^c(eT)$ for some constants $C,c<\infty$. 

For the finite-sample result, we need a batch-level version of the population
projection objects that determine the limiting conjecture matrix. Recall that
$\tilde Z_i^{t,k}:=Z_i^t-I_k^{-1}\sum_{s\in\mathcal T_k}Z_i^s$ is seller $i$'s
centered experimentation variable in batch $k$. For seller $i$, write
$\tilde Z_{\mathcal R_i}^{t,k}:=(\tilde Z_r^{t,k}:r\in\mathcal R_i)$ for the
vector of centered experimentation variables corresponding to the prices that
seller $i$ observes and includes in her regression. The coordinates of this
vector are indexed by $\mathcal R_i$; thus, for a vector $\mathbf x$ indexed by
$\mathcal R_i$, $[\mathbf x]_r$ denotes the coordinate associated with seller
$r\in\mathcal R_i$.

Define the empirical covariance matrix of seller $i$'s observed experimentation
variables in batch $k$ by $\widehat\Sigma_i^k:=I_k^{-1}\sum_{t\in\mathcal T_k}
\tilde Z_{\mathcal R_i}^{t,k}(\tilde Z_{\mathcal R_i}^{t,k})^\top$. For an
unobserved rival $\ell\in\mathcal U_i$, define the empirical covariance vector
between seller $\ell$'s omitted experimentation and seller $i$'s observed
experimentation variables by $\widehat{\mathbf{c}}_{i\ell}^k:=I_k^{-1}\sum_{t\in\mathcal
T_k}\tilde Z_{\mathcal R_i}^{t,k}\tilde Z_\ell^{t,k}$. When
$\widehat\Sigma_i^k$ is nonsingular, let
$\pi_{i\ell}^k:=(\widehat\Sigma_i^k)^{-1}\widehat{\mathbf{c}}_{i\ell}^k$; if
$\widehat\Sigma_i^k$ is singular, set $\pi_{i\ell}^k=0$. Thus
$\pi_{i\ell}^k$ is the batch-$k$ empirical linear-projection coefficient from
projecting the omitted rival's experimentation $\tilde Z_\ell^{t,k}$ on the
experimentation variables observed by seller $i$. In particular,
$[\pi_{i\ell}^k]_i$ measures the part of rival $\ell$'s omitted experimentation
that is predicted by seller $i$'s own experimentation, after controlling for the
other prices observed by seller $i$.

For $\Delta \geq 0$, let $\mathcal E_k(\Delta)$ be the event that, in batch $k$, the
observed experimental design is well conditioned and these empirical projection
coefficients are within $\Delta$ of their population limits:
\begin{equation}\label{eq:good_event_experiment}
\mathcal E_k(\Delta)
:=
\left\{
\min_i\lambda_{\min}(\widehat\Sigma_i^k)
\ge
\frac12\min_i\lambda_{\min}(\Sigma_i^\star)
\right\}
\cap
\left\{
\max_{\substack{i,\ r\in\mathcal R_i,\ \ell\in\mathcal U_i}}
\left|
\frac{\delta_\ell^k}{\delta_r^k}[\pi_{i\ell}^k]_r
-
[\pi_{i\ell}^\star]_r
\right|
\le \Delta
\right\},
\end{equation}
where $\lambda_{\min}(B)$ denotes the smallest eigenvalue of  a symmetric matrix $B$.

The first event rules out near-collinearity among the price variations included
in seller $i$'s regression. The second event says that the batch-level
omitted-variable-bias coefficients are close to their limiting values. If the maximum is over an empty set, the second event is
interpreted as the whole sample space.

We also define 
$
\gamma
:=
\max_i
\left\{
1-u_i
+
u_i
\sup_{\mathbf p\in\mathcal P}
\left\|Dz^{(A^{\mathcal G,\star})}(\mathbf p)\right\|_\infty
\right\}
$
which is smaller than $1$ under the conditions of Theorem \ref{thm:main_theorem}. 

The proof of
Theorem~\ref{thm:cv_rate_general} is deferred to the Appendix.

\begin{theorem}[Finite-sample rate]\label{thm:cv_rate_general}
Suppose all assumptions of Theorem~\ref{thm:main_theorem} hold,
and let $\mathbf p^\star:=\mathbf p^{(A^{\mathcal G,\star})}$ be the limiting
fixed point. Suppose, in addition, that there exist deterministic sequences
$\Delta^k=\widetilde O(I_k^{-1/4})$ and
$\tilde r_k=\widetilde O(I_k^{-1/2})$ such that, for all sufficiently large
$k$, $\mathbb P(\mathcal E_k(\Delta^k)^c)\le\tilde r_k$.

Assume common experimentation magnitudes
$\delta_i^k=\delta^k=(\log(eI_k)/I_k)^{1/4}$ for all $i\in[n]$.\footnote{Theorem~\ref{thm:main_theorem} only requires
$\delta_j^k/\delta_i^k\to1$ for all $i,j\in[n]$, so the experimentation
magnitudes vanish at a common asymptotic rate. For the finite-sample rate, we
impose the simpler normalization $\delta_i^k=\delta_j^k$ for all sellers. The
same argument extends to unequal magnitudes if the convergence of
$\delta_j^k/\delta_i^k$ to $1$ is itself controlled at the required finite-batch
rate.}
If $I_k\asymp b^k$ for some $1<b\le\gamma^{-4}$, then, for all $T\ge1$,
\[
\mathbb E\big[\|\mathbf p(T)-\mathbf p^\star\|_\infty^2\big]
=
\widetilde O(T^{-1/2}).
\]
\end{theorem}

\begin{remark}
\label{rem:finite-batch-stabilization-sufficient}
The finite-batch condition
\(\mathbb P(\mathcal E_k(\Delta^k)^c)\le\tilde r_k\) for all sufficiently large
\(k\) is a standard concentration requirement on the experimentation design and
is not restrictive in general. It is satisfied, for
example, when, within each batch and conditional on the past, the
experimentation vectors $\{Z^t:t\in\mathcal T_k\}$ are i.i.d., bounded,
independent of the demand shocks, and have common first and second moments whose
induced covariance matrices satisfy the positive-definiteness condition in
Assumption~\ref{assu:A-star}.\footnote{Indeed, in this case, because $n$ is
fixed, entrywise Hoeffding bounds for the empirical first and second moments
imply that, with probability at least $1-CI_k^{-2}$,
\[
\max_i\|\widehat\Sigma_i^k-\Sigma_i^\star\|_\infty
+
\max_{\substack{i,\ \ell\in\mathcal U_i}}
\|\widehat{\mathbf c}_{i\ell}^k- \mathbf c_{i\ell}^\star\|_\infty
\le
C\sqrt{\frac{\log(eI_k)}{I_k}} .
\]
Since the dimension is fixed, the same bound also controls
$\|\widehat\Sigma_i^k-\Sigma_i^\star\|_2$ up to a constant. Weyl's inequality
therefore gives
$
\lambda_{\min}(\widehat\Sigma_i^k)
\ge
\lambda_{\min}(\Sigma_i^\star)
-
\|\widehat\Sigma_i^k-\Sigma_i^\star\|_2
\ge
\frac12\min_j\lambda_{\min}(\Sigma_j^\star)
$ 
for all sellers $i$ and all sufficiently large $k$ on this event. Thus the
empirical observed-design matrices are uniformly nonsingular with high
probability.

On the same event, the usual matrix-inverse stability bound gives, uniformly
over $i$ and $\ell\in\mathcal U_i$,
\[
\|(\widehat\Sigma_i^k)^{-1}\widehat{\mathbf c}_{i\ell}^k
-
(\Sigma_i^\star)^{-1}\mathbf c_{i\ell}^\star\|_\infty
\le
C\sqrt{\frac{\log(eI_k)}{I_k}} .
\]
Under the common experimentation magnitudes imposed in
Theorem~\ref{thm:cv_rate_general}, $\delta_\ell^k/\delta_r^k=1$. Hence the
finite-batch stabilization event holds with
$\Delta^k=C\sqrt{\log(eI_k)/I_k}$ and $\tilde r_k=CI_k^{-2}$ for suitable
constants $C<\infty$. }
\end{remark}

\begin{remark}
The rate in Theorem~\ref{thm:cv_rate_general} is a fixed-market rate. That is, the number of sellers, the feedback structure, the demand system, and
the limiting experimentation covariance objects are held fixed. Thus the theorem describes
the dependence on the time horizon \(T\), not a uniform rate over \(n\).
An \(n\)-dependent rate would require specifying a sequence of markets indexed
by \(n\), including how the feedback, demand
curvature, covariance matrices, conditioning constants, and contraction modulus
scale with \(n\). Without such additional structure, there is no intrinsic
dependence on \(n\). 
\end{remark}

\section{Demand Function Examples}\label{sec:examples}

In this section we provide examples and conditions where
Assumption~\ref{assu:ldl-contraction}, the contraction condition needed for
Theorem~\ref{thm:main_theorem}, holds. We analyze two standard demand
specifications: Linear and Multinomial Logit (MNL). The discussion is useful
both for the Nash case, which arises when $A^{\mathcal G,\star}=0$, and for the
general case with nonzero conjectures. The latter can arise when sellers have
partial feedback and unobserved rivals' experimentation is correlated with the
price variation included in the regression.

Throughout this section, write $A=A^{\mathcal G,\star}$ and
$\beta_i(\mathbf p)=\beta_i^{(A)}(\mathbf p)$. We have 
\begin{equation}\label{eq:z-simpler-general-A}
z_i^{(A)}(\mathbf p)
=
\frac12 p_i+\frac12\frac{\lambda_i(\mathbf p)}{\beta_i(\mathbf p)} .
\end{equation}
Differentiating \eqref{eq:z-simpler-general-A} and using the quotient rule gives
\begin{align}
\frac{\partial z_i^{(A)}}{\partial p_j}(\mathbf p)
&=
\frac12\mathbf 1_{\{i=j\}}
+
\frac12\frac{\partial}{\partial p_j}
\left(\frac{\lambda_i}{\beta_i}\right)(\mathbf p) \nonumber\\
&=
\frac12\mathbf 1_{\{i=j\}}
+
\frac12
\frac{
\beta_i(\mathbf p)\partial_{p_j}\lambda_i(\mathbf p)
-
\lambda_i(\mathbf p)\partial_{p_j}\beta_i(\mathbf p)
}{
\beta_i(\mathbf p)^2
}.
\label{eq:dz-general-A}
\end{align}
Split the second term into a ``frozen $\beta$'' first-derivative part and a
curvature part:
\begin{equation}\label{eq:split-general-A}
\frac{
\beta_i(\mathbf p)\partial_{p_j}\lambda_i(\mathbf p)
-
\lambda_i(\mathbf p)\partial_{p_j}\beta_i(\mathbf p)
}{
\beta_i(\mathbf p)^2
}
=
\underbrace{\frac{\partial_{p_j}\lambda_i(\mathbf p)}{\beta_i(\mathbf p)}}_{\text{competition and conjectural first-order effects}}
-
\underbrace{
\frac{\lambda_i(\mathbf p)}{\beta_i(\mathbf p)^2}
\partial_{p_j}\beta_i(\mathbf p)
}_{\text{curvature and slope-rotation effects}} .
\end{equation}
Define $L^{\mathrm{comp},A}(\mathbf p)$ and
$L^{\mathrm{curv},A}(\mathbf p)$ entrywise by
\begin{equation}\label{eq:Lcomp-Lcurv-general-A}
L^{\mathrm{comp},A}_{ij}(\mathbf p)
:=
\frac12\mathbf 1_{\{i=j\}}
+
\frac12\frac{\partial_{p_j}\lambda_i(\mathbf p)}{\beta_i(\mathbf p)},
\qquad
L^{\mathrm{curv},A}_{ij}(\mathbf p)
:=
-\frac12
\frac{\lambda_i(\mathbf p)}{\beta_i(\mathbf p)^2}
\partial_{p_j}\beta_i(\mathbf p).
\end{equation}
Then \eqref{eq:dz-general-A}--\eqref{eq:Lcomp-Lcurv-general-A} yield the exact
decomposition
\begin{equation}\label{eq:Jac-decomp-general-A}
Dz^{(A)}(\mathbf p)
=
L^{\mathrm{comp},A}(\mathbf p)+L^{\mathrm{curv},A}(\mathbf p).
\end{equation}

This decomposition characterizes the two sources of sensitivity in the limiting
LDL update. The matrix $L^{\mathrm{comp},A}$ captures how the fitted target
moves when the demand level changes but the learned own-price slope is held
fixed. Its off-diagonal entries reflect cross-price effects, and hence the
strategic interaction among sellers. In the Nash case $A=0$, the diagonal
entries of $L^{\mathrm{comp},A}$ are zero, and the off-diagonal entries reduce
to the familiar cross-effect terms. Under full feedback, these off-diagonal
terms are not omitted-variable bias: rivals' prices are observed and controlled
for directly, but the fitted target can still move with rivals' prices. Under
partial feedback, nonzero entries of $A$ additionally reflect the conjectural
slope created by correlated experimentation with unobserved rivals. The matrix
$L^{\mathrm{curv},A}$ captures the contribution of demand curvature: its entries
depend on $\partial_{p_j}\beta_i$, the rate at which the learned own-price slope
changes with the price vector. In a linear demand model, this curvature term
vanishes entirely.\footnote{
More precisely, the diagonal entries $L^{\mathrm{curv},A}_{ii}$ capture the
relative concavity of each seller's demand, analogous to the $\alpha$-convexity
metrics used to quantify the degree of convexity \cite{light2021family}. The
off-diagonal entries $L^{\mathrm{curv},A}_{ij}$ depend on mixed partial
derivatives and on the conjecture row $A_{i\cdot}$ through
$\partial_{p_j}\beta_i$. While these terms clearly involve competitors, we
classify them as ``curvature'' effects because they quantify how rival prices
rotate the relevant demand slope, whereas $L^{\mathrm{comp},A}$ captures how
they shift the demand level holding the slope fixed. Thus, the off-diagonal
curvature error is small when rival price changes shift demand without
significantly altering the learned own-price slope.
}

This decomposition explicitly characterizes why Assumption~\ref{assu:ldl-contraction}
is a stability condition. The fitted linear model is local and possibly
misspecified. Stability requires that the first-order competitive effects and
the curvature-induced changes in the learned slope do not make the fitted target
too sensitive to the current price vector. We now show that this condition has
simple forms in standard models.

\noindent\textbf{Linear Demand Model}: Under linear demands,
\begin{equation}\label{eq:linear_model}
\lambda_i(\mathbf p)=a_i-b_{ii}p_i+\sum_{j\ne i}b_{ij}p_j,
\qquad
b_{ii}>0,\quad b_{ij}\ge0,\quad j\ne i .
\end{equation}
For this model, Assumption~\ref{assu:cv-foc-sufficiency} holds whenever the
learned own-price slope is positive (see Lemma~\ref{lm:best_reply_satisfying_models}
in the Appendix). Moreover, $\beta_i^{(A)}=b_{ii}-\sum_{j\ne i}A_{ij}b_{ij}$ is
constant in $\mathbf p$, so $L^{\mathrm{curv},A}=0$. Hence there is no curvature
error. Lemma~\ref{lem:jacobian-z-linear-general} in the Appendix shows that, for
any conjecture matrix $A$ with $\beta_i^{(A)}>0$,
\[
\|Dz^{(A)}\|_\infty
=
\max_i
\frac{
\left|\sum_{j\ne i}A_{ij}b_{ij}\right|+\sum_{j\ne i}b_{ij}
}{
2\left(b_{ii}-\sum_{j\ne i}A_{ij}b_{ij}\right)
}.
\]
In the nonnegative conjectures case, the contraction condition
becomes
\[
\sum_{j\ne i}(1+3A_{ij})b_{ij}<2b_{ii},
\qquad i\in[n].
\]
For $A=0$, this reduces to
$\frac12\sum_{j\ne i}b_{ij}<b_{ii}$ for every seller $i$ which is a familiar diagonal dominance condition used to prove the stability of Nash equilibrium in linear demand models. Thus convergence is
ensured when own-price sensitivity dominates the aggregate cross-price effects.

\noindent\textbf{MNL Demand Model}: Consider the multinomial logit demand
\begin{equation}\label{eq:mnl_linear}
\lambda_i(\mathbf p)
=
\frac{\exp(a_i-b_ip_i)}
{1+\sum_{j=1}^n\exp(a_j-b_jp_j)},
\qquad i\in[n].
\end{equation}
This model satisfies Assumption~\ref{assu:cv-foc-sufficiency} for any conjecture
matrix $A$ (see Lemma~\ref{lm:best_reply_satisfying_models} in the Appendix).
For MNL demand, the learned own-price slope under conjecture $A$ is
$\beta_i^{(A)}(\mathbf p)=\lambda_i(\mathbf p)
\{b_i(1-\lambda_i(\mathbf p))-\sum_{j\ne i}A_{ij}b_j\lambda_j(\mathbf p)\}$.
Lemma~\ref{lem:mnl_contraction} gives exact Jacobian formulas for arbitrary
$A$ whenever this learned own-price slope is positive. These formulas can be
checked directly for any given set of parameters and feedback-induced
conjectures.

A simple sufficient condition is especially transparent in the Nash case
$A=0$. If price sensitivities are symmetric, $b_j=b$ for all $j$, then
Lemma~\ref{lem:mnl_contraction} shows that
$\sup_{\mathbf p\in\mathcal P}\lambda_i(\mathbf p)<3/5$ for every seller $i$
implies $\sup_{\mathbf p\in\mathcal P}\|Dz^{(0)}(\mathbf p)\|_\infty<1$.
More generally, even with heterogeneous price sensitivities and nonzero
conjectures, maintaining sufficiently small market shares is the key force behind contraction. Intuitively, as a seller's
market share grows toward dominance, cross-price and curvature effects become
large relative to the own-price slope, making the fitted target more sensitive
to competitors' prices and potentially destabilizing the learning dynamics.

\section{Beyond Convergence}
\label{sec:beyond-contraction}

In this section, we discuss what happens to the learning dynamics if the global
stability condition (Assumption~\ref{assu:ldl-contraction}) fails. Contraction
provides a clean sufficient condition for convergence, and such conditions are
typically used in games to ensure that learning dynamics converge. We show,
however, that the failure of this condition does not by itself mean that the
learning algorithm becomes economically meaningless. Instead, the algorithm
remains connected to the key conjecture matrix determined by the experimentation
and information structure across sellers, even though the learning procedure is
misspecified in two ways: sellers may ignore the prices of unobserved
competitors, and they fit a local linear demand model to a nonlinear,
multi-agent demand system.

In particular, we show that the direction of the algorithm's unprojected price adjustment is componentwise aligned with the
true CV marginal revenue under the conjecture
matrix selected by the experimentation and feedback structure.

Let $z^{(A)}=(z_1^{(A)},\ldots,z_n^{(A)})$, let
$U:=\operatorname{diag}(u_i)$, and define 
\begin{equation}\label{eq:Operator_F_A_Star}
F^{(A)}(\mathbf p)
:=
\mathrm{proj}_{\mathcal P}\left((I-U)\mathbf p+Uz^{(A)}(\mathbf p)\right).
\end{equation}
Thus $F^{(A)}$ is the deterministic update map obtained by replacing the fitted
OLS coefficients in the LDL algorithm with their population limits.

Let $A^\star:=A^{\mathcal G,\star}$ and define seller $i$'s true CV marginal revenue under the induced conjecture
row $A^\star_{i\cdot}$ as
\begin{equation}
\label{eq:cv_marginal_revenue}
\mathcal M_i^{(A^\star)}(\mathbf p)
:=
\nabla r_i(\mathbf p)^\top v_i(A^\star_{i\cdot})
=
\lambda_i(\mathbf p)
+
p_i
\left(
\partial_{p_i}\lambda_i(\mathbf p)
+
\sum_{j\ne i}A_{ij}^\star
\partial_{p_j}\lambda_i(\mathbf p)
\right).
\end{equation}
Here $v_i(A^\star_{i\cdot})=e_i+\sum_{j\ne i}A_{ij}^\star e_j$ is the
conjectured local direction introduced in Section~\ref{Sec:CV}. As we discussed in Section \ref{sec:convergence}, 
the omitted-variable bias generated by correlated experimentation makes seller
$i$ behave as if a change in her own price is evaluated along this direction.

Hence, using the definitions from Section~\ref{sec:convergence},
$\beta_i^{(A^\star)}(\mathbf p)$ is the own-price slope that seller $i$ learns
in the population limit under the conjecture matrix $A^\star$, and
$z_i^{(A^\star)}(\mathbf p)$ is the corresponding revenue-maximizing target
price under the fitted local linear demand model,  we have 
\begin{align}
z_i^{(A^\star)}(\mathbf p)-p_i
&=
\frac{
\lambda_i(\mathbf p)+p_i\beta_i^{(A^\star)}(\mathbf p)
-2p_i\beta_i^{(A^\star)}(\mathbf p)
}
{2\beta_i^{(A^\star)}(\mathbf p)}
\nonumber\\
&=
\frac{
\lambda_i(\mathbf p)-p_i\beta_i^{(A^\star)}(\mathbf p)
}
{2\beta_i^{(A^\star)}(\mathbf p)}
=
\frac{1}{2\beta_i^{(A^\star)}(\mathbf p)}
\mathcal M_i^{(A^\star)}(\mathbf p).
\label{eq:fundamental_identity}
\end{align}

Equation~\eqref{eq:fundamental_identity} is the key link between the statistical
learning rule and the induced (limiting) CV game. The algorithm moves the focal price
toward $z_i^{(A^\star)}(\mathbf p)$. Hence, before projection and statistical
error, the deterministic update satisfies 
\[
p_i^{\mathrm{new}}-p_i
=
u_i\left(z_i^{(A^\star)}(\mathbf p)-p_i\right)
=
\frac{u_i}{2\beta_i^{(A^\star)}(\mathbf p)}
\mathcal M_i^{(A^\star)}(\mathbf p).
\]
Assumption~\ref{assu:A-star} guarantees
\(\beta_i^{(A^\star)}(\mathbf p)\ge \beta_{\min}>0\) on \(\mathcal P\), so the multiplier
$u_i/(2\beta_i^{(A^\star)}(\mathbf p))$ is strictly positive. Thus, seller $i$'s
price increases exactly when her true CV marginal revenue is positive and decreases
exactly when her true CV marginal revenue is negative. 

This is true even though sellers neither know the conjecture matrix $A^\star$
nor explicitly optimize the CV objective, and even though they fit a
misspecified local linear demand model. The observation shows why the algorithm
remains meaningful beyond the contraction regime: asymptotically,
the sign of seller $i$'s price adjustment coincides with the sign of her true
CV marginal revenue under $A^\star$, while the magnitude of the adjustment is
scaled by the learned own-price slope 
$\beta_i^{(A^\star)}$.

In addition, using \eqref{eq:fundamental_identity}, the update map
$F^{(A^\star)}$ from \eqref{eq:Operator_F_A_Star} can be written as
\begin{equation}
\label{eq:operator_gradient_identity}
F^{(A^\star)}(\mathbf p)
=
\mathrm{proj}_{\mathcal P}
\left(
\mathbf p
+
U\mathcal D^{(A^\star)}(\mathbf p)
\mathcal M^{(A^\star)}(\mathbf p)
\right).
\end{equation}
where
$\mathcal M^{(A^\star)}(\mathbf p):=(\mathcal M_i^{(A^\star)}(\mathbf p))_{i\in[n]}$
and 
\[
\mathcal D^{(A^\star)}(\mathbf p)
:=
\operatorname{diag}
\left(
\frac{1}{2\beta_i^{(A^\star)}(\mathbf p)}
\right)_{i\in[n]} .
\]
Thus $F^{(A^\star)}$ is the projected version of a seller-specific rescaling of
the true CV marginal-revenue adjustment.
This connects the dynamics directly to the primitives emphasized in the
paper. The feedback structure and the correlation structure of
experimentation determine the conjecture matrix $A^\star$. The conjecture
matrix $A^\star$ determines the CV marginal-revenue vector
$\mathcal M^{(A^\star)}$. The algorithm then applies a positive diagonal rescaling of the CV marginal-revenue vector. Consequently, correlated experimentation
and partial feedback remain economically meaningful even when global convergence
is not guaranteed: they determine not only the CV equilibria that may be reached
under contraction, but also the CV adjustment dynamics tracked by the algorithm
outside the contraction regime.

The same identity also clarifies the equilibrium interpretation. If
$\mathbf p^\star\in\mathrm{int}(\mathcal P)$ is a fixed point of
$F^{(A^\star)}$, then the projection is inactive and
$z^{(A^\star)}(\mathbf p^\star)=\mathbf p^\star$. By
\eqref{eq:fundamental_identity}, this is equivalent to
$\mathcal M_i^{(A^\star)}(\mathbf p^\star)=0$ for every seller $i$, which is
exactly the CV first-order system. Under
Assumption~\ref{assu:cv-foc-sufficiency}, such an interior fixed point is a
CV$(A^\star)$ equilibrium.

We now formalize the tracking property. We show that the stochastic
focal-price sequence is an asymptotic pseudo-orbit of
$F^{(A^\star)}$. This means that the one-step statistical approximation error
vanishes almost surely. Therefore, if prices do not converge to a single point,
their long-run behavior is still governed by the deterministic CV adjustment map
in \eqref{eq:operator_gradient_identity}, rather than by persistent estimation
noise.
\begin{theorem}
\label{thm:apt_tracking}
Suppose Assumption~\ref{assu:A-star} holds, but
Assumption~\ref{assu:ldl-contraction} does not necessarily hold. Consider the
update rule in Algorithm~\ref{alg:partial-feedback-ldl}. 
Assume common experimentation magnitudes
$\delta_i^k=\delta^k=(\log(eI_k)/I_k)^{1/4},\,\forall i\in[n]$. Suppose there
exist deterministic sequences $\Delta^k=\widetilde O(I_k^{-1/4})$ and
$\tilde r_k=\widetilde O(I_k^{-1/2})$ such that, for all sufficiently large
$k$, the finite-batch stabilization event satisfies
$\mathbb P(\mathcal E_k(\Delta^k)^c)\le \tilde r_k$, where
$\mathcal E_k(\Delta^k)$ is defined in Section~\ref{sec:converge_rate}.
Finally, suppose $I_k\asymp b^k$ for some $b>1$.

Then, almost surely, the stochastic focal-price sequence
$\{\hat{\mathbf p}^k\}_{k\ge1}$ forms an asymptotic pseudo-orbit of the
discrete-time dynamical system induced by $F^{(A^{\mathcal G,\star})}$. That is,\footnote{Because $F^{(A^{\mathcal G,\star})}$ is continuous, maps the
compact set $\mathcal P$ into itself, and the focal prices lie in
$\mathcal P$, the sequence is precompact. Hence the standard asymptotic
pseudo-orbit result for discrete-time maps applies. By \citet[Definition 2.3
and Lemma 2.3]{hirsch2001chain}, the omega-limit set of
$\{\hat{\mathbf p}^k\}$ is nonempty, compact, invariant, and internally chain
transitive for $F^{(A^{\mathcal G,\star})}$ on the almost-sure event on which
the sequence is an asymptotic pseudo-orbit. In words, any nonconvergent
long-run behavior must be recurrent behavior of the deterministic CV adjustment
map. }
\begin{equation}
\label{eq:apt_limit}
\lim_{k\to\infty}
\left\|
\hat{\mathbf p}^{k+1}
-
F^{(A^{\mathcal G,\star})}(\hat{\mathbf p}^k)
\right\|_\infty
=
0
\qquad\text{a.s.}
\end{equation}
\end{theorem}

Theorem~\ref{thm:apt_tracking} formalizes the interpretation of the
algorithm beyond the contraction regime. Without contraction, the deterministic
CV adjustment map may have multiple equilibria, periodic cycles, or more general
recurrent invariant sets. The theorem says that the stochastic learning process
tracks this deterministic adjustment map asymptotically.

In particular, away from binding projection constraints, the tracking statement
and \eqref{eq:operator_gradient_identity} imply that, for large $k$,
\[
\hat p_i^{k+1}-\hat p_i^k
=
\frac{u_i}{2\beta_i^{(A^\star)}(\hat{\mathbf p}^k)}
\mathcal M_i^{(A^\star)}(\hat{\mathbf p}^k)
+
o(1)
\qquad \text{a.s.}
\]
Thus the algorithmic price adjustments follow the marginal-revenue adjustment
dynamics of the true demand system under the conjectures induced by the
correlation structure of experimentation and the feedback structure.

\section{Conclusion}

This paper studies competitive dynamic pricing when multiple sellers learn demand using partial feedback, running local price experiments and repeatedly refitting a misspecified linear demand model using only the prices they observe. Our main contribution is to show that the long-run outcome of such practical experimentation and learning is pinned down by a systematic learning bias induced by the interaction between the market's feedback structure and the correlation structure of experimentation. In particular, Theorem~\ref{thm:main_theorem} establishes that, under certain conditions on demand, the dynamics converge to a CV$(A^{\mathcal{G},\star})$ equilibrium, where the conjecture matrix $A^{\mathcal{G},\star}$ is endogenously determined by which rivals a seller monitors and how their respective experimentation is correlated. Here, CV$(A^{\mathcal{G},\star})$ denotes a conjectural variations equilibrium in which each seller best responds under an implicit local conjecture about how rivals'
prices co-move with its own. 
This result provides a learning foundation for conjectural variations that does not posit conjectures as behavioral primitives: sellers behave \emph{as if} unobserved rivals co-move with their own prices because their demand estimates are biased. 

Complementing our results, Theorem~\ref{thm:cv_rate_general} provides a finite-sample guarantee, showing that with appropriate scaling of parameters, the mean squared price error decays on the order of $\widetilde{O}(T^{-1/2})$. In addition, the comparative statics in Proposition~\ref{Prop:Price_increase} show that under strategic complementarity, higher induced conjectures lead to higher prices. 

Several directions are promising for future work. First, it would be valuable to move beyond locally linear fitted demand models to analyze learning dynamics under non-linear learning algorithms. As discussed in Section~\ref{sec:intuition-beyond-ldl}, we expect that other learning rules that estimate marginal demand from observed data  will inherit a similar conjectural component. Second, while our paper is theoretical in nature, it would be interesting to conduct empirical work that maps actual market
feedback structures and experimentation correlations, and quantifies how these
objects relate to observed price dynamics and market outcomes.

\bibliographystyle{apalike}
\bibliography{bib_arxivVersion.bib}

\newpage
\appendix

\section{Proofs of Theorem~\ref{thm:main_theorem}, Theorem \ref{thm:cv_rate_general}, Theorem \ref{thm:apt_tracking}, and Corollary \ref{corr:nash}}
\label{sec:proof-thm-partial-feedback-convergence}
\textbf{Notation.}
Write $A^\star:=A^{\mathcal G,\star}$,
$B_{\max}:=\max_i B_i$, and
$D_{\mathcal P}:=\|\mathbf p^h-\mathbf p^l\|_\infty$. All constants below are
deterministic, finite, independent of $k$, $K$, and $T$, and may change from
line to line. For any finite index set $J\subseteq[n]$, vectors indexed by $J$
are written in increasing numerical order. For a vector $\mathbf x_J=(x_j:j\in J)$ indexed by a finite set $J\subseteq[n]$, $[\mathbf x]_j$
denotes the coordinate associated with index $j\in J$; similarly, $e_j$ denotes
the corresponding coordinate vector in $\mathbb R^{|J|}$. 

For each batch $k$, let $\mathcal T_k=\{t_{k-1}+1,\ldots,t_k\}$,
$\bar p_j^k:=I_k^{-1}\sum_{t\in\mathcal T_k}p_j^t$, and for any \(J\subseteq[n]\), let
\(\bar{\mathbf p}_J^k:=(\bar p_j^k:j\in J)\),  
$\bar D_i^k:=I_k^{-1}\sum_{t\in\mathcal T_k}D_i^t$,
$\bar Z_i^k:=I_k^{-1}\sum_{t\in\mathcal T_k}Z_i^t$, and
$\tilde Z_i^{t,k}:=Z_i^t-\bar Z_i^k$. For seller $i$, set
$\tilde Z_{\mathcal R_i}^{t,k}:=(\tilde Z_r^{t,k})_{r\in\mathcal R_i}$ and
define
\begin{equation}\label{eq:partial-empirical-covariances}
\widehat\Sigma_i^k
:=
I_k^{-1}\sum_{t\in\mathcal T_k}
\tilde Z_{\mathcal R_i}^{t,k}
(\tilde Z_{\mathcal R_i}^{t,k})^\top,
\qquad
\widehat{\mathbf{c}}_{i\ell}^k
:=
I_k^{-1}\sum_{t\in\mathcal T_k}
\tilde Z_{\mathcal R_i}^{t,k}\tilde Z_\ell^{t,k},
\quad \ell\in\mathcal U_i .
\end{equation}
Let $\Gamma_i^k:=\operatorname{diag}(\delta_r^k:r\in\mathcal R_i)$. Since
$\mathbf p_{\mathcal R_i}^t-\bar{\mathbf p}_{\mathcal R_i}^k
=\Gamma_i^k\tilde Z_{\mathcal R_i}^{t,k}$, the centered OLS Gram matrix
satisfies $\widehat G_i^k=\Gamma_i^k\widehat\Sigma_i^k\Gamma_i^k$. Notice that $\delta_i^k>0,\,\forall i\in[n],\, k\ge 1$. Hence
$\widehat G_i^k$ is nonsingular if and only if $\widehat\Sigma_i^k$ is
nonsingular.

Whenever $\widehat\Sigma_i^k$ is nonsingular, set
$\pi_{i\ell}^k:=(\widehat\Sigma_i^k)^{-1}\widehat{\mathbf{c}}_{i\ell}^k$ for
$\ell\in\mathcal U_i$; otherwise set $\pi_{i\ell}^k=0$. This convention makes
$\pi_{i\ell}^k$ measurable on every sample path. Its deterministic limit is
$\pi_{i\ell}^\star:=(\Sigma_i^\star)^{-1}\mathbf c_{i\ell}^\star$. Also set
$s_{\min}:=\min_i\lambda_{\min}(\Sigma_i^\star)>0$.

For $r\in\mathcal R_i$, define
\[
\psi_{ir}^{\mathcal G,\star}(\mathbf p)
:=
\partial_{p_r}\lambda_i(\mathbf p)
+
\sum_{\ell\in\mathcal U_i}
[\pi_{i\ell}^\star]_r\,\partial_{p_\ell}\lambda_i(\mathbf p).
\]
Set
\[
\beta_i^{\mathcal G,\star}(\mathbf p):=
-\psi_{ii}^{\mathcal G,\star}(\mathbf p),
\qquad
\theta_{ij}^{\mathcal G,\star}(\mathbf p):=
\psi_{ij}^{\mathcal G,\star}(\mathbf p),\quad j\in\mathcal S_i,
\]
and
\[
\alpha_i^{\mathcal G,\star}(\mathbf p)
:=
\lambda_i(\mathbf p)
-
\sum_{r\in\mathcal R_i}\psi_{ir}^{\mathcal G,\star}(\mathbf p)p_r
=
\lambda_i(\mathbf p)
+
\beta_i^{\mathcal G,\star}(\mathbf p)p_i
-
\sum_{j\in\mathcal S_i}
\theta_{ij}^{\mathcal G,\star}(\mathbf p)p_j .
\]
For $\mathbf p_{\mathcal S_i}=(p_j:j\in\mathcal S_i)$, write
$h_i^{\mathcal S_i}(\alpha_i,\beta_i,\theta_i; \mathbf p_{\mathcal S_i})
:=(\alpha_i+\sum_{j\in\mathcal S_i}\theta_{ij}p_j)/(2\beta_i)$. By
\eqref{eq:A-star-partial-feedback} and \eqref{eq:beta_alpha},
$\beta_i^{\mathcal G,\star}(\mathbf p)=\beta_i^{(A^\star)}(\mathbf p)$ and,
for every $\mathbf p\in\mathcal P$,
\begin{equation}\label{eq:shared-target-identity}
h_i^{\mathcal S_i}
(\alpha_i^{\mathcal G,\star}(\mathbf p),
\beta_i^{\mathcal G,\star}(\mathbf p),
\theta_i^{\mathcal G,\star}(\mathbf p);\mathbf p_{\mathcal S_i})
=
z_i^{(A^\star)}(\mathbf p).
\end{equation}

For $k\ge1$, define
\begin{equation}\label{eq:shared-Fm}
F_k^{(A^\star)}(\mathbf p)
:=
\mathrm{proj}_{\mathcal P^k}
\left((I-U)\mathbf p+Uz^{(A^\star)}(\mathbf p)\right),
\end{equation}
and let
$F^{(A^\star)}(\mathbf p)
:=\mathrm{proj}_{\mathcal P}((I-U)\mathbf p+Uz^{(A^\star)}(\mathbf p))$.

Finally, define
$\bar\varepsilon_i^k:=I_k^{-1}\sum_{t\in\mathcal T_k}\varepsilon_i^t$ and
$W_{ir}^k:=I_k^{-1}\sum_{t\in\mathcal T_k}
\tilde Z_r^{t,k}\varepsilon_i^t$ for $r\in\mathcal R_i$. Throughout,
$a_k:=(c_a\log(eI_k))^{1/2}$, where $c_a>0$ is chosen large enough in the
concentration step.

\begin{lemma}\label{lem:partial-projection-convergence}
Suppose Assumption~\ref{assu:A-star} holds and
$\delta_j^k/\delta_r^k\to1$ for all $j,r\in[n]$. Then, for every seller $i$,
$r\in\mathcal R_i$, and $\ell\in\mathcal U_i$,
\[
\frac{\delta_\ell^k}{\delta_r^k}[\pi_{i\ell}^k]_r
\xrightarrow{p}
[\pi_{i\ell}^\star]_r.
\]
Consequently, there exist deterministic sequences
$\Delta_{ir\ell}^k\downarrow0$ such that
\[
\mathbb P\left(
\left|
\frac{\delta_\ell^k}{\delta_r^k}[\pi_{i\ell}^k]_r
-
[\pi_{i\ell}^\star]_r
\right|
>
\Delta_{ir\ell}^k
\right)\to0 .
\]
\end{lemma}

\begin{proof}
By Assumption~\ref{assu:A-star},
$\widehat\Sigma_i^k\to\Sigma_i^\star$ and
$\widehat{\mathbf c}_{i\ell}^k\to  \mathbf c_{i\ell}^\star$ almost surely. Since
$\Sigma_i^\star$ is positive definite, $\widehat\Sigma_i^k$ is nonsingular
eventually almost surely. On this eventual full-probability event,
\[
\pi_{i\ell}^k
=
(\widehat\Sigma_i^k)^{-1}\widehat{\mathbf c}_{i\ell}^k
\to
(\Sigma_i^\star)^{-1}\mathbf c_{i\ell}^\star
=
\pi_{i\ell}^\star
\qquad a.s.
\]
The deterministic ratio $\delta_\ell^k/\delta_r^k\to1$, so the product rule
gives the claimed convergence in probability.

It remains to choose deterministic tolerances. Let
$X_{ir\ell}^k:=|(\delta_\ell^k/\delta_r^k)[\pi_{i\ell}^k]_r-[\pi_{i\ell}^\star]_r|$.
Since $X_{ir\ell}^k\to0$ in probability, for every $m\ge1$ there is $K_m$ such that
$\mathbb P(X_{ir\ell}^k>1/m)\le1/m$ for all $k\ge K_m$. Increasing the $K_m$'s if
necessary, assume $K_m<K_{m+1}$. Define $\Delta_{ir\ell}^k:=1$ for $k<K_1$ and
$\Delta_{ir\ell}^k:=1/m$ for $K_m\le k<K_{m+1}$. Then
$\Delta_{ir\ell}^k\downarrow0$ and
$\mathbb P(X_{ir\ell}^k>\Delta_{ir\ell}^k)\to0$.
\end{proof}

\begin{proof}[Proof of Theorem~\ref{thm:main_theorem}]
The algorithm preserves feasibility. Indeed, if
$\hat p_i^k\in\mathcal P_i^k$, then
$p_i^t=\hat p_i^k+\delta_i^kZ_i^t\in[p_i^l,p_i^h]$ because
$|Z_i^t|\le B_i$. Since $\delta_i^k\downarrow0$,
$\mathcal P_i^k\subseteq\mathcal P_i^{k+1}$, and the projected update gives
$\hat p_i^{k+1}\in\mathcal P_i^{k+1}$. Hence
$\hat{\mathbf p}^k\in\mathcal P^k\subseteq\mathcal P$ for all $k$.

Let $\delta_{\min}^k:=\min_i\delta_i^k$,
$\delta_{\max}^k:=\max_i\delta_i^k$, and
\[
\Delta^k:=
\max_{\substack{i\in[n],\,r\in\mathcal R_i,\,\ell\in\mathcal U_i}}
\Delta_{ir\ell}^k,
\]
with $\Delta^k:=0$ if the maximum is over an empty set. Since there are only
finitely many triples, $\Delta^k\to0$. Together with the set definition in \eqref{eq:good_event_experiment}, define

\begin{equation}\label{eq:good-event-partial-feedback}
\begin{aligned}
\mathcal A_k
:=
\mathcal{E}_k(\Delta^k)\cap
\left\{
|\bar\varepsilon_i^k|\le a_kI_k^{-1/2},
\quad
|W_{ir}^k|\le a_kI_k^{-1/2},
\quad
\forall i,\ r\in\mathcal R_i
\right\}.
\end{aligned}
\end{equation}

If the maximum in \(\mathcal E_k(\Delta^k)\) is over an empty set, the corresponding event is interpreted as the whole sample space. 

By the twice continuous differentiability of the demand functions and the
positive-slope condition in Assumption~\ref{assu:A-star},
$z^{(A^\star)}$ is $C^1$ on $\mathcal P$. Let
$\rho:=\sup_{\mathbf p\in\mathcal P}\|Dz^{(A^\star)}(\mathbf p)\|_\infty$. By Assumption~\ref{assu:ldl-contraction}, $\rho<1$.
For any $\mathbf x,\mathbf y\in\mathcal P$, non-expansiveness of projection and
the mean-value theorem give
\begin{equation}\label{eq:partial-F-contraction}
\begin{aligned}
\|F^{(A^\star)}(\mathbf x)-F^{(A^\star)}(\mathbf y)\|_\infty
&\le
\|(I-U)(\mathbf x-\mathbf y)
+
U[z^{(A^\star)}(\mathbf x)-z^{(A^\star)}(\mathbf y)]\|_\infty
\\
&\le
\gamma\|\mathbf x-\mathbf y\|_\infty,
\end{aligned}
\end{equation}
where $\gamma:=\max_i\{1-u_i+u_i\rho\}<1$. Thus $F^{(A^\star)}$ is a
contraction on the complete metric space $\mathcal P$. Let $\mathbf p^\star$
denote its unique fixed point.

We decompose
\begin{equation}\label{eq:partial-main-decomposition}
\begin{aligned}
\mathbb E\|\hat{\mathbf p}^{k+1}-\mathbf p^\star\|_\infty
=
&\mathbb E\left[
\|\hat{\mathbf p}^{k+1}-\mathbf p^\star\|_\infty 1_{\mathcal A_k}
\right]
\\
&+
\mathbb E\left[
\|\hat{\mathbf p}^{k+1}-\mathbf p^\star\|_\infty 1_{\mathcal A_k^c}
\right].
\end{aligned}
\end{equation}

\textbf{Part I: the good event.}
Fix seller $i$ and work on $\mathcal A_k$. Since $\mathbf p_{\mathcal R_i}^t-\bar{\mathbf p}_{\mathcal R_i}^k
=\Gamma_i^k\tilde Z_{\mathcal R_i}^{t,k}$,
\begin{equation}\label{eq:partial-gram}
\sum_{t\in\mathcal T_k}
(\mathbf p_{\mathcal R_i}^t-\bar{\mathbf p}_{\mathcal R_i}^k)
(\mathbf p_{\mathcal R_i}^t-\bar{\mathbf p}_{\mathcal R_i}^k)^\top
=
I_k\Gamma_i^k\widehat\Sigma_i^k\Gamma_i^k .
\end{equation}
On $\mathcal A_k$, this matrix is nonsingular, so the singular-design fallback
is not used.

Let $\delta^k=(\delta_1^k,\ldots,\delta_n^k)$. For $t\in\mathcal T_k$, define
\[
R_i^t
:=
\lambda_i(\hat{\mathbf p}^k+\operatorname{diag}(\delta^k)Z^t)
-
\lambda_i(\hat{\mathbf p}^k)
-
\sum_{\ell=1}^n
\partial_{p_\ell}\lambda_i(\hat{\mathbf p}^k)\delta_\ell^k Z_\ell^t .
\]
Since $\lambda_i$ is $C^2$ on compact $\mathcal P$, the line segment between
$\hat{\mathbf p}^k$ and
$\hat{\mathbf p}^k+\operatorname{diag}(\delta^k)Z^t$ lies in $\mathcal P$, and
$Z^t$ is uniformly bounded. Hence
\begin{equation}\label{eq:partial-remainder-bound}
|R_i^t|\le C_1(\delta_{\max}^k)^2
\qquad \forall i,t,k ,
\end{equation}
for some $C_1>0$.

Set
$\widehat R_i^k:=I_k^{-1}\sum_{t\in\mathcal T_k}
\tilde Z_{\mathcal R_i}^{t,k}R_i^t$ and
$\widehat W_i^k:=(W_{ir}^k)_{r\in\mathcal R_i}$.

Using centered regressors, the OLS slope vector
$\hat b_i^{k+1}:=(\hat b_{ir}^{k+1}:r\in\mathcal R_i)$ satisfies
\[
\hat b_i^{k+1}
=
\left(
\sum_{t\in\mathcal T_k}
(\mathbf p_{\mathcal R_i}^t-\bar{\mathbf p}_{\mathcal R_i}^k)
(\mathbf p_{\mathcal R_i}^t-\bar{\mathbf p}_{\mathcal R_i}^k)^\top
\right)^{-1}
\sum_{t\in\mathcal T_k}
(\mathbf p_{\mathcal R_i}^t-\bar{\mathbf p}_{\mathcal R_i}^k)(D_i^t-\bar D_i^k).
\]
The centered regressors sum to zero, so the same formula is obtained if
$D_i^t-\bar D_i^k$ is replaced by $D_i^t$. Also,
$I_k^{-1}\sum_t\tilde Z_{\mathcal R_i}^{t,k}Z_r^t
=\widehat\Sigma_i^k e_r$ for $r\in\mathcal R_i$, and
$I_k^{-1}\sum_t\tilde Z_{\mathcal R_i}^{t,k}Z_\ell^t
=\widehat{\mathbf{c}}_{i\ell}^k$ for $\ell\in\mathcal U_i$, where $e_r$ is the coordinate
vector corresponding to index $r$. Then one has that
\[\sum_{t\in\mathcal T_k}
( \mathbf p_{\mathcal R_i}^t-\bar{\mathbf p}_{\mathcal R_i}^k)(D_i^t-\bar D_i^k)=I_k\Gamma_i^k\left(\widehat\Sigma_i^k\Gamma_i^k(\partial_{p_r}\lambda_i(\hat{\mathbf p}^k))_{r\in\mathcal R_i}+\sum_{\ell\in\mathcal U_i}
\delta_\ell^k\widehat{\mathbf{c}}_{i\ell}^k
\partial_{p_\ell}\lambda_i(\hat{\mathbf p}^k)+\widehat R_i^k+\widehat W_i^k\right).\]

Combining these identities with \eqref{eq:partial-gram} gives
\[
\begin{aligned}
\hat b_i^{k+1}
=(\partial_{p_r}\lambda_i(\hat{\mathbf p}^k))_{r\in\mathcal R_i}
+
(\Gamma_i^k)^{-1}
\Bigg[
\sum_{\ell\in\mathcal U_i}
\delta_\ell^k\pi_{i\ell}^k
\partial_{p_\ell}\lambda_i(\hat{\mathbf p}^k)
+
(\widehat\Sigma_i^k)^{-1}
(\widehat R_i^k+\widehat W_i^k)
\Bigg].
\end{aligned}
\]
Therefore, for $r\in\mathcal R_i$,
\[
\begin{aligned}
\hat b_{ir}^{k+1}
-
\psi_{ir}^{\mathcal G,\star}(\hat{\mathbf p}^k)
=
&\sum_{\ell\in\mathcal U_i}
\left(
\frac{\delta_\ell^k}{\delta_r^k}[\pi_{i\ell}^k]_r
-
[\pi_{i\ell}^\star]_r
\right)
\partial_{p_\ell}\lambda_i(\hat{\mathbf p}^k)
\\
&+
\left[
(\Gamma_i^k)^{-1}
(\widehat\Sigma_i^k)^{-1}
(\widehat R_i^k+\widehat W_i^k)
\right]_r .
\end{aligned}
\]
On $\mathcal A_k$, the first line is bounded by $C_2\Delta^k$ for some $C_2>0$. Also,
$\lambda_{\min}(\widehat\Sigma_i^k)\ge s_{\min}/2$ and
$|\tilde Z_r^{t,k}|\le2B_r$ imply, since $n$ is fixed, that
$\|\widehat R_i^k\|_\infty\le B_{\max}C_1(\delta_{\max}^k)^2$,
$\|\widehat W_i^k\|_\infty\le a_kI_k^{-1/2}$, and
$\|(\widehat\Sigma_i^k)^{-1}\|_\infty\le C_3$. The uniform boundedness of
$\delta_j^k/\delta_i^k$ gives
$(\delta_{\max}^k)^2/\delta_{\min}^k\le C_4\delta_{\max}^k$. Thus, on
$\mathcal A_k$,
\begin{equation}\label{eq:partial-coeff-error}
|\hat b_{ir}^{k+1}
-
\psi_{ir}^{\mathcal G,\star}(\hat{\mathbf p}^k)|
\le C_5\eta_k
\qquad
\forall i,\ r\in\mathcal R_i .
\end{equation}

where $\eta_k:=\Delta^k+\delta_{\max}^k+a_kI_k^{-1/2}/\delta_{\min}^k$. The theorem's
scaling implies $I_k\to\infty$ and, since $n$ is fixed,
$\delta_{\min}^k I_k^{1/2}/\log(eI_k)\to\infty$. Hence $\eta_k\to0$.

In particular,
$|\hat\beta_i^{k+1}-\beta_i^{\mathcal G,\star}(\hat{\mathbf p}^k)|\le C_5\eta_k$
and
$|\hat\theta_{ij}^{k+1}-\theta_{ij}^{\mathcal G,\star}(\hat{\mathbf p}^k)|
\le C_5\eta_k$ for $j\in\mathcal S_i$.

Next, because the regression includes an intercept,
$\hat\alpha_i^{k+1}
=\bar D_i^k-\sum_{r\in\mathcal R_i}\hat b_{ir}^{k+1}\bar p_r^k$. The Taylor
expansion gives
\[
\bar D_i^k
=
\lambda_i(\hat{\mathbf p}^k)
+
\sum_{\ell=1}^n
\partial_{p_\ell}\lambda_i(\hat{\mathbf p}^k)\delta_\ell^k\bar Z_\ell^k
+
\bar R_i^k+\bar\varepsilon_i^k,
\]
where $\bar R_i^k:=I_k^{-1}\sum_{t\in\mathcal T_k}R_i^t$ and
$|\bar R_i^k|\le C_1(\delta_{\max}^k)^2$. Since
$\bar p_r^k=\hat p_r^k+\delta_r^k\bar Z_r^k$,
\[
\begin{aligned}
\hat\alpha_i^{k+1}
-
\alpha_i^{\mathcal G,\star}(\hat{\mathbf p}^k)
=
&\sum_{\ell=1}^n
\partial_{p_\ell}\lambda_i(\hat{\mathbf p}^k)
\delta_\ell^k\bar Z_\ell^k
-
\sum_{r\in\mathcal R_i}
\psi_{ir}^{\mathcal G,\star}(\hat{\mathbf p}^k)
\delta_r^k\bar Z_r^k
\\
&-
\sum_{r\in\mathcal R_i}
\bigl(\hat b_{ir}^{k+1}
-
\psi_{ir}^{\mathcal G,\star}(\hat{\mathbf p}^k)\bigr)
(\hat p_r^k+\delta_r^k\bar Z_r^k)
+
\bar R_i^k+\bar\varepsilon_i^k .
\end{aligned}
\]
The quantities $\bar Z_r^k$ are uniformly bounded,
$\hat{\mathbf p}^k\in\mathcal P$, and the limiting coefficients are bounded on
compact $\mathcal P$. Using \eqref{eq:partial-coeff-error}, and absorbing
$|\bar\varepsilon_i^k|\le a_kI_k^{-1/2}$ into
$a_kI_k^{-1/2}/\delta_{\min}^k$ because $\delta_{\min}^k\le1$ eventually,
we obtain, on $\mathcal A_k$ and for all sufficiently large $k$,
\begin{equation}\label{eq:partial-alpha-error}
|\hat\alpha_i^{k+1}
-
\alpha_i^{\mathcal G,\star}(\hat{\mathbf p}^k)|
\le C_6\eta_k .
\end{equation}

By Assumption~\ref{assu:A-star},
\(\beta_i^{\mathcal G,\star}(\mathbf p)\ge \beta_{\min}\) on $\mathcal P$. Since
$\eta_k\to0$, \eqref{eq:partial-coeff-error} and
Assumption~\ref{assu:A-star} imply that, for all
sufficiently large $k$, on $\mathcal A_k$ the small-slope fallback is not used
and \(\hat\beta_i^{k+1}>\beta_{\min}/2>\underline\beta_i\). The singular-design fallback
has already been ruled out by \eqref{eq:partial-gram} and $\mathcal E_{k}(\Delta^k)$.

On the compact coefficient set reached for all sufficiently large $k$ on
$\mathcal A_k$, with \(\beta\ge \beta_{\min}/2\) and
$\mathbf p_{\mathcal S_i}\in\mathcal P_{\mathcal S_i}$, the derivatives of
$h_i^{\mathcal S_i}$ are bounded. Moreover,
$\|\bar{\mathbf p}_{\mathcal S_i}^k-\hat{\mathbf p}_{\mathcal S_i}^k\|_\infty
\le B_{\max}\delta_{\max}^k$. Therefore
\eqref{eq:partial-coeff-error}, \eqref{eq:partial-alpha-error}, and
\eqref{eq:shared-target-identity} imply
\[
|\hat z_i^{k+1}
-
z_i^{(A^\star)}(\hat{\mathbf p}^k)|
\le C_7\eta_k .
\]
Using non-expansiveness of coordinatewise projection onto rectangles,
\[
\left\|
\mathrm{proj}_{\mathcal P^{k+1}}
\left((I-U)\hat{\mathbf p}^k+U\hat z^{k+1}\right)
-
F_{k+1}^{(A^\star)}(\hat{\mathbf p}^k)
\right\|_\infty
\le C_8\eta_k .
\]
For every $\mathbf x\in\mathbb R^n$ and every $k$,
$\|\mathrm{proj}_{\mathcal P^k}(\mathbf x)
-\mathrm{proj}_{\mathcal P}(\mathbf x)\|_\infty
\le B_{\max}\delta_{\max}^{k}$. Hence
\[
\|F_{k+1}^{(A^\star)}(\hat{\mathbf p}^k)
-F^{(A^\star)}(\hat{\mathbf p}^k)\|_\infty
\le B_{\max}\delta_{\max}^{k+1}\le B_{\max}\delta_{\max}^{k}.
\]
After increasing the constant if needed, the actual update satisfies
\begin{equation}\label{eq:partial-one-step-to-F}
\|\hat{\mathbf p}^{k+1}
-
F^{(A^\star)}(\hat{\mathbf p}^k)\|_\infty
\le C_{\eta}\eta_k
\qquad\text{on }\mathcal A_k .
\end{equation}
Combining \eqref{eq:partial-one-step-to-F} with the contraction
\eqref{eq:partial-F-contraction} gives, on $\mathcal A_k$ and for all
sufficiently large $k$,
\begin{equation}\label{eq:partial-good-recursion}
\|\hat{\mathbf p}^{k+1}-\mathbf p^\star\|_\infty
\le
\gamma\|\hat{\mathbf p}^{k}-\mathbf p^\star\|_\infty
+
v_k,
\qquad
v_k:=C_v\eta_k .
\end{equation}
Since $\eta_k\to0$, $v_k\to0$.

\textbf{Part II: the bad event has vanishing probability.}
We prove $\mathbb P(\mathcal A_k^c)\to0$. First,
$\mathbb P(\mathcal E_k(\Delta^k)^c)\to0$ by
Assumption~\ref{assu:A-star}, Lemma~\ref{lem:partial-projection-convergence}, continuity of eigenvalues, and
positive definiteness of $\Sigma_i^\star$. 

It remains to control the shock terms. Fix $i$ and define
$\mathcal C_k:=\mathcal F_{t_{k-1}}\vee\sigma(Z^t:t\in\mathcal T_k)$. By the
over-time independence of the demand shocks and the conditional independence of
the batch experimentation variables from the batch shocks, conditional on
$\mathcal C_k$ the variables $\{\varepsilon_i^t:t\in\mathcal T_k\}$ are
independent, mean zero, and have the same local log-MGF bound as
unconditionally. Let $\{\mu_t:t\in\mathcal T_k\}$ be $\mathcal C_k$-measurable
weights with $|\mu_t|\le L$. Since the shocks have zero mean and finite log-MGF
in a neighborhood of zero, and since $n$ is fixed, there are
$K_\varepsilon<\infty$ and $s_0>0$ such that
$\log\mathbb E[e^{s\varepsilon_i^t}]\le K_\varepsilon s^2$ for all
$|s|\le s_0$, uniformly over $i$. Thus, for $s\in(0,s_0/L)$,
\[
\mathbb E\left[
\exp\left(s\sum_{t\in\mathcal T_k}\mu_t\varepsilon_i^t\right)
\ \middle|\ \mathcal C_k
\right]
\le
\exp(K_\varepsilon L^2s^2I_k).
\]
Therefore, for $x>0$ and $s\in(0,s_0/L)$,
\[
\mathbb P\left(
I_k^{-1}\sum_{t\in\mathcal T_k}\mu_t\varepsilon_i^t>x
\ \Big|\ \mathcal C_k
\right)
\le
\exp\{-sI_kx+K_\varepsilon L^2s^2I_k\}.
\]
Taking $x=a_kI_k^{-1/2}$ and $s=x/(2K_\varepsilon L^2)$ is valid for all
sufficiently large $k$, and gives
\begin{equation}\label{eq:partial-noise-concentration}
\mathbb P\left(
\left|
I_k^{-1}\sum_{t\in\mathcal T_k}\mu_t\varepsilon_i^t
\right|
>
a_kI_k^{-1/2}
\ \Big|\ \mathcal C_k
\right)
\le
2\exp(-c_La_k^2)
\end{equation}
for some $c_L>0$, where the lower tail is obtained by applying the same bound
to $-\mu_t$.

Apply \eqref{eq:partial-noise-concentration} with $\mu_t=1$ to control
$\bar\varepsilon_i^k$, and with $\mu_t=\tilde Z_r^{t,k}$ to control $W_{ir}^k$.
Since $|\tilde Z_r^{t,k}|\le2B_r\le2B_{\max}$, choosing \(c_a\) sufficiently large, depending only on
the shock tail constants and \(B_{\max}\), makes the right-hand side of
\eqref{eq:partial-noise-concentration} of order \(I_k^{-2}\) in both
cases. Hence, 
\[
\mathbb P(|\bar\varepsilon_i^k|>a_kI_k^{-1/2})\le CI_k^{-2},
\qquad
\mathbb P(|W_{ir}^k|>a_kI_k^{-1/2})\le CI_k^{-2},
\]
for some constant $C>0$.

A finite union bound over $i$ and $r\in\mathcal R_i$, together with the two
preceding probability bounds, yields $\mathbb P(\mathcal A_k^c)\to0$.

Since $\hat{\mathbf p}^{k+1},\mathbf p^\star\in\mathcal P$,
\begin{equation}\label{eq:partial-bad-event}
\mathbb E\left[
\|\hat{\mathbf p}^{k+1}-\mathbf p^\star\|_\infty1_{\mathcal A_k^c}
\right]
\le
D_{\mathcal P}\mathbb P(\mathcal A_k^c)
\to0 .
\end{equation}

\textbf{Part III: convergence in expectation.}
Using \eqref{eq:partial-good-recursion},
\[
\mathbb E\left[
\|\hat{\mathbf p}^{k+1}-\mathbf p^\star\|_\infty1_{\mathcal A_k}
\right]
\le
\gamma\mathbb E\|\hat{\mathbf p}^{k}-\mathbf p^\star\|_\infty+v_k .
\]
Combining this with \eqref{eq:partial-main-decomposition} and
\eqref{eq:partial-bad-event}, for all sufficiently large $k$,
\[
\mathbb E\|\hat{\mathbf p}^{k+1}-\mathbf p^\star\|_\infty
\le
\gamma\mathbb E\|\hat{\mathbf p}^{k}-\mathbf p^\star\|_\infty+w_k,
\]
where $w_k:=v_k+D_{\mathcal P}\mathbb P(\mathcal A_k^c)$. After changing
finitely many initial terms if necessary, $w_k$ is bounded and $w_k\to0$. Let
$x_k:=\mathbb E\|\hat{\mathbf p}^{k}-\mathbf p^\star\|_\infty$, and choose
$k_0$ such that the recursion holds for every $k\ge k_0$. Iterating gives, for
$k\ge k_0$,
\[
x_{k+1}
\le
\gamma^{k+1-k_0}x_{k_0}
+
\sum_{m=k_0}^{k}\gamma^{k-m}w_m .
\]
The first term converges to zero. For the convolution term, fix
$\varepsilon>0$ and choose $M\ge k_0$ such that
$w_m\le\varepsilon(1-\gamma)/2$ for all $m\ge M$. The finite sum over
$m<M$ vanishes as $k\to\infty$, while the tail is at most $\varepsilon/2$.
Thus $x_{k+1}\to0$, and equivalently
\[
\mathbb E\left[
\left\|\hat{\mathbf p}^{k}
-
\mathbf p^{\star}\right\|_\infty
\right]\to0 .
\]

Finally, suppose
$\mathbf p^\star=\mathbf p^{(A^{\mathcal G,\star})}\in\mathrm{int}(\mathcal P)$.
Since $\mathbf p^\star$ is a fixed point of $F^{(A^\star)}$ and the projection
onto $\mathcal P$ returns an interior point, the projection is inactive at the
fixed point. Therefore
$(I-U)\mathbf p^\star+Uz^{(A^\star)}(\mathbf p^\star)=\mathbf p^\star$. Since
$u_i\in(0,1]$ for all $i$, $U$ is invertible, and hence
$z^{(A^\star)}(\mathbf p^\star)=\mathbf p^\star$. Coordinatewise,
\[
\frac{
\lambda_i(\mathbf p^\star)
+
\beta_i^{(A^\star)}(\mathbf p^\star)p_i^\star
}{
2\beta_i^{(A^\star)}(\mathbf p^\star)
}
=
p_i^\star .
\]
Because $\beta_i^{(A^\star)}(\mathbf p^\star)>0$, this is equivalent to
\[
\lambda_i(\mathbf p^\star)
+
p_i^\star
\left(
\partial_{p_i}\lambda_i(\mathbf p^\star)
+
\sum_{j\ne i}
A_{ij}^{\mathcal G,\star}\partial_{p_j}\lambda_i(\mathbf p^\star)
\right)
=0,
\qquad \forall\,i\in[n].
\]
These are the CV$(A^{\mathcal G,\star})$ first-order conditions. By
Assumption~\ref{assu:cv-foc-sufficiency}, they characterize
CV$(A^{\mathcal G,\star})$ best replies. Hence $\mathbf p^\star$ is a
CV$(A^{\mathcal G,\star})$ equilibrium.
\end{proof}

\begin{proof}[Proof of Theorem~\ref{thm:cv_rate_general}]
The shared notation from Section~\ref{sec:proof-thm-partial-feedback-convergence}
is used throughout.

\textbf{Step 1: Good-event estimates and bad-event probability bound.}
Throughout this proof, when we refer to the good event $\mathcal A_k$ defined in
\eqref{eq:good-event-partial-feedback}, the tolerance $\Delta^k$ is understood
to be the finite-rate sequence appearing in the statement of
Theorem~\ref{thm:cv_rate_general}.

Recall the good event $\mathcal{A}_k$ definition in \eqref{eq:good-event-partial-feedback}, and the bound definition

\[
\eta_k:=
\Delta^k+\delta_{\max}^k+\frac{a_kI_k^{-1/2}}{\delta_{\min}^k}.
\]
Since $\Delta^k=\widetilde O(I_k^{-1/4})$, the common experimentation magnitudes satisfy
$\delta_i^k=\delta^k=(\log(eI_k)/I_k)^{1/4},\,\forall i\in[n]$, implying $\delta_{\min}^k=\delta_{\max}^k$, 
and
$a_k\asymp\sqrt{\log(eI_k)}$, we have
\begin{equation}\label{eq:rate-eta-order}
\eta_k=\widetilde O(I_k^{-1/4}).
\end{equation}

On $\mathcal A_k$, the finite-batch event $\mathcal E_k(\Delta^k)$ gives
$\min_i\lambda_{\min}(\widehat\Sigma_i^k)\ge s_{\min}/2$ and
\[
\left|
\frac{\delta_\ell^k}{\delta_r^k}[\pi_{i\ell}^k]_r
-
[\pi_{i\ell}^\star]_r
\right|
\le \Delta^k
\quad
\forall i,\ r\in\mathcal R_i,\ \ell\in\mathcal U_i .
\]

Combining \eqref{eq:partial-one-step-to-F} and \eqref{eq:partial-good-recursion}, we have that
\begin{equation}\label{eq:rate-good-recursion}
\|\hat{\mathbf p}^{k+1}-\mathbf p^\star\|_\infty
\le
\gamma\|\hat{\mathbf p}^{k}-\mathbf p^\star\|_\infty+v_k
\qquad\text{on }\mathcal A_k,
\end{equation}
where $v_k:=C_v\eta_k=\widetilde{O}(I_k^{-1/4})$ for all sufficiently large $k$.

By assumption,
$\mathbb P(\mathcal E_k(\Delta^k)^c)\le\tilde r_k$ for all sufficiently large
$k$. It remains to control the noise part of $\mathcal A_k$. The conditional
concentration argument leading to \eqref{eq:partial-noise-concentration} applies
with the same conditioning
$\mathcal C_k=\mathcal F_{t_{k-1}}\vee\sigma(Z^t:t\in\mathcal T_k)$, because
the batch experimentation variables are conditionally independent of the batch
shocks. Applying \eqref{eq:partial-noise-concentration} with $\mu_t=1$ controls
$\bar\varepsilon_i^k$, and applying it with $\mu_t=\tilde Z_r^{t,k}$ controls
$W_{ir}^k$. Since $|\tilde Z_r^{t,k}|\le2B_r\le2B_{\max}$, choosing $c_a$
large enough and taking a finite union bound over $i$ and $r\in\mathcal R_i$
gives
\[
\mathbb P\left(
\max_i|\bar\varepsilon_i^k|>a_kI_k^{-1/2}
\ \text{or}\
\max_{i,\ r\in\mathcal R_i}|W_{ir}^k|>a_kI_k^{-1/2}
\right)
\le CI_k^{-2}.
\]
Thus, for all sufficiently large $k$,
\begin{equation}\label{eq:rate-bad-prob}
\mathbb P(\mathcal A_k^c)\le \tilde r_k+CI_k^{-2}.
\end{equation}

\textbf{Step 2: an $L^2$ recursion at batch endpoints.}
By the feasibility argument in the proof of
Theorem~\ref{thm:main_theorem}, both
$\hat{\mathbf p}^{k+1}$ and $\mathbf p^\star$ lie in $\mathcal P$. Hence, on
$\mathcal A_k^c$,
$\|\hat{\mathbf p}^{k+1}-\mathbf p^\star\|_\infty\le D_{\mathcal P}$. Combining
this bound with \eqref{eq:rate-good-recursion}, for all sufficiently large $k$,
\[
\|\hat{\mathbf p}^{k+1}-\mathbf p^\star\|_\infty
\le
\left(\gamma\|\hat{\mathbf p}^{k}-\mathbf p^\star\|_\infty+v_k\right)
1_{\mathcal A_k}
+
D_{\mathcal P}1_{\mathcal A_k^c}.
\]
Taking $L^2$ norms and using Minkowski's inequality gives
\[
\left(\mathbb E\|\hat{\mathbf p}^{k+1}-\mathbf p^\star\|_\infty^2\right)^{1/2}
\le
\gamma
\left(\mathbb E\|\hat{\mathbf p}^{k}-\mathbf p^\star\|_\infty^2\right)^{1/2}
+
\bar w_k,
\]
where
$\bar w_k:=v_k+D_{\mathcal P}\mathbb P(\mathcal A_k^c)^{1/2}=C_v\eta_k+D_{\mathcal P}\mathbb P(\mathcal A_k^c)^{1/2}$. By
\eqref{eq:rate-bad-prob}, $\sqrt{x+y}\le\sqrt x+\sqrt y$, and
\eqref{eq:rate-eta-order},
\[
\bar w_k
\le
C\left(
\Delta^k+\delta^k+
\frac{a_kI_k^{-1/2}}{\delta^k}
+
\sqrt{\tilde r_k}
+
I_k^{-1}
\right)
=
\widetilde O(I_k^{-1/4}).
\]
Let
$y_k:=(\mathbb E\|\hat{\mathbf p}^{k}-\mathbf p^\star\|_\infty^2)^{1/2}$.
Choose $k_0$ large enough that the preceding recursion holds for every
$k\ge k_0$. Then, for every $K\ge k_0$,
\begin{equation}\label{eq:rate-convolution-L2}
y_{K+1}
\le
\gamma^{K+1-k_0}y_{k_0}
+
\sum_{m=k_0}^{K}\gamma^{K-m}\bar w_m .
\end{equation}
Since all focal prices lie in $\mathcal P$, $y_{k_0}\le D_{\mathcal P}$.

Because $I_k\asymp b^k$, $I_K=\Theta(t_K)$. Moreover, for
$0\le j\le K-k_0$, $I_{K-j}\asymp I_K/b^j$. The bound
$\bar w_k=\widetilde O(I_k^{-1/4})$ gives, after increasing the
polylogarithmic constant if needed,
$\bar w_{K-j}\le\widetilde C\, b^{j/4}I_K^{-1/4}$ for some $\widetilde C>0$. Therefore
\[
\sum_{m=k_0}^{K}\gamma^{K-m}\bar w_m
=
\sum_{j=0}^{K-k_0}\gamma^j\bar w_{K-j}
\le
\widetilde C I_K^{-1/4}
\sum_{j=0}^{K-k_0}(\gamma b^{1/4})^j .
\]
Since $b\le\gamma^{-4}$, $\gamma b^{1/4}\le1$. If
$\gamma b^{1/4}<1$, the last sum is bounded by a constant; if
$\gamma b^{1/4}=1$, it is $O(K)$, which is absorbed by the
$\widetilde O(\cdot)$ notation because $K=O(\log t_K)$. Hence
\[
\sum_{m=k_0}^{K}\gamma^{K-m}\bar w_m
=
\widetilde O(I_K^{-1/4})
=
\widetilde O(t_K^{-1/4}).
\]
Also, $b\le\gamma^{-4}$ implies $\gamma^K=O(I_K^{-1/4})$, and therefore
$\gamma^{K+1-k_0}y_{k_0}=\widetilde O(t_K^{-1/4})$. Substituting these two
bounds into \eqref{eq:rate-convolution-L2} gives
\[
\left(
\mathbb E
\left\|
\hat{\mathbf p}^{K+1}-\mathbf p^\star
\right\|_\infty^2
\right)^{1/2}
=
\widetilde O(t_K^{-1/4}).
\]
Squaring both sides yields
\begin{equation}\label{eq:rate-endpoint}
\mathbb E
\left\|
\hat{\mathbf p}^{K+1}-\mathbf p^\star
\right\|_\infty^2
=
\widetilde O(t_K^{-1/2}).
\end{equation}

\textbf{Step 3: arbitrary periods.}
Fix $T\ge1$ and let $K(T):=\max\{K\ge0:t_K<T\}$. For all sufficiently large
$T$, $K(T)\ge k_0$. Since $t_{K(T)}<T\le t_{K(T)+1}$, period $T$ lies in batch
$K(T)+1$, and hence
\[
\mathbf p(T)=\hat{\mathbf p}^{K(T)+1}+\delta^{K(T)+1}Z^T .
\]
Because $|Z_i^T|\le B_i$,
$\|\mathbf p(T)-\hat{\mathbf p}^{K(T)+1}\|_\infty
\le B_{\max}\delta^{K(T)+1}$. Therefore,
\[
\mathbb E\|\mathbf p(T)-\mathbf p^\star\|_\infty^2
\le
2\mathbb E\|\hat{\mathbf p}^{K(T)+1}-\mathbf p^\star\|_\infty^2
+
2B_{\max}^2(\delta^{K(T)+1})^2 .
\]
Since $I_k\asymp b^k$ with fixed $b>1$, both
$t_{K(T)}=\Theta(T)$ and $I_{K(T)+1}=\Theta(T)$ for all sufficiently large $T$.
Using \eqref{eq:rate-endpoint},
\[
\mathbb E\|\hat{\mathbf p}^{K(T)+1}-\mathbf p^\star\|_\infty^2
=
\widetilde O(T^{-1/2}).
\]
Moreover,
\[
(\delta^{K(T)+1})^2
=
\left(\frac{\log(eI_{K(T)+1})}{I_{K(T)+1}}\right)^{1/2}
=
\widetilde O(T^{-1/2}).
\]
Combining the last three displays gives
\[
\mathbb E\left[
\|\mathbf p(T)-\mathbf p^\star\|_\infty^2
\right]
=
\widetilde O(T^{-1/2})
\]
for all sufficiently large $T$. Enlarging the implicit constant handles the
finitely many remaining periods, completing the proof.
\end{proof}

\begin{proof}[Proof of Theorem~\ref{thm:apt_tracking}]
Throughout this proof, when we refer to the good event $\mathcal A_k$ defined in
\eqref{eq:good-event-partial-feedback}, the tolerance $\Delta^k$ is understood
to be the finite-rate sequence appearing in the statement of
Theorem~\ref{thm:apt_tracking}. We use the one-step estimate already established in the proof of Theorem~\ref{thm:main_theorem} and Theorem~\ref{thm:cv_rate_general}.  In particular, the proof of \eqref{eq:partial-one-step-to-F} shows that, on the event
$\mathcal A_k$ defined in \eqref{eq:good-event-partial-feedback},
\[
\left\|
\hat{\mathbf p}^{k+1}
-
F^{(A^\star)}(\hat{\mathbf p}^{k})
\right\|_\infty
\le C_{\eta}\eta_k,
\]
for all sufficiently large $k$, where $\eta_k$ is defined at the start of the proof of 
Theorem \ref{thm:cv_rate_general}.

Indeed, the derivation of \eqref{eq:partial-one-step-to-F} does not use Assumption~\ref{assu:ldl-contraction}. 

Importantly, the derivation of this one-step bound
uses only the Taylor approximation, the OLS concentration within batch $k$, the
finite-batch stabilization event, the positivity of the limiting learned slopes
from Assumption~\ref{assu:A-star}, and the shrinking experimentation magnitude.
It does not use the contraction property
Assumption~\ref{assu:ldl-contraction}. In the proof of
Theorem~\ref{thm:main_theorem}, contraction is used only after
\eqref{eq:partial-one-step-to-F}, to turn the one-step approximation bound into a
recursive bound around the fixed point.

Under the present assumptions,
$\eta_k=\widetilde O(I_k^{-1/4})$. Since $I_k\asymp b^k$ with $b>1$,
$\eta_k\to0$. Moreover, the bad-event estimate \eqref{eq:rate-bad-prob} gives
\[
\mathbb P(\mathcal A_k^c)
\le
\tilde r_k+CI_k^{-2}
=
\widetilde O(I_k^{-1/2}).
\]
Because $I_k\asymp b^k$, the sequence
$\{\mathbb P(\mathcal A_k^c)\}_{k\ge1}$ is summable. Hence, by the first
Borel-Cantelli lemma, almost surely there exists a finite random integer
$K(\omega)$ such that $\mathcal A_k$ occurs for every $k\ge K(\omega)$.

On this almost-sure event, \eqref{eq:partial-one-step-to-F} holds for all
sufficiently large $k$. Since $\eta_k\to0$, we obtain
\[
\lim_{k\to\infty}
\left\|
\hat{\mathbf p}^{k+1}
-
F^{(A^\star)}(\hat{\mathbf p}^k)
\right\|_\infty
=
0
\qquad\text{a.s.}
\]
This is exactly the definition of an asymptotic pseudo-orbit of the continuous
map $F^{(A^\star)}$ on the metric space $\mathcal P$. Continuity of
$F^{(A^\star)}$ follows from the continuity of $\lambda_i$, the positivity of
$\beta_i^{(A^\star)}$ on $\mathcal P$ under Assumption~\ref{assu:A-star}, and
the continuity of projection onto the compact rectangle $\mathcal P$.

Finally, suppose that Assumption~\ref{assu:cv-foc-sufficiency} holds for
$A=A^\star$ and that, on a given sample path,
$\hat{\mathbf p}^k\to\bar{\mathbf p}\in\mathrm{int}(\mathcal P)$. Then also
$\hat{\mathbf p}^{k+1}\to\bar{\mathbf p}$. By continuity of $F^{(A^\star)}$ and
the vanishing one-step error in \eqref{eq:apt_limit}, we have
$\bar{\mathbf p}=F^{(A^\star)}(\bar{\mathbf p})$. Since
$\bar{\mathbf p}$ is interior, the projection is inactive at the fixed point.
Because $u_i>0$ for all $i$, $U$ is invertible, and hence
$z^{(A^\star)}(\bar{\mathbf p})=\bar{\mathbf p}$. By
\eqref{eq:fundamental_identity},
$\mathcal M_i^{(A^\star)}(\bar{\mathbf p})=0$ for every seller $i$. These are
the CV$(A^\star)$ first-order conditions, and by
Assumption~\ref{assu:cv-foc-sufficiency}, $\bar{\mathbf p}$ is a
CV$(A^\star)$ equilibrium.
\end{proof}

\begin{proof}[Proof of Corollary~\ref{corr:nash}]
By definition of the information-dependent conjecture matrix,
$A_{ij}^{\mathcal G,\star}=0$ for every observed rival
$j\in\mathcal S_i$. Now consider an unobserved rival
$\ell\in\mathcal U_i$. Since $\mathbf c_{i\ell}^\star=0$, the limiting partial
linear-projection coefficient vector is
$\pi_{i\ell}^\star=(\Sigma_i^\star)^{-1}\mathbf c_{i\ell}^\star=0$. Hence
$A_{i\ell}^{\mathcal G,\star}=[\pi_{i\ell}^\star]_i=0$. Therefore
$A^{\mathcal G,\star}=0$.

The convergence claim now follows directly from
Theorem~\ref{thm:main_theorem}, with
$A^{\mathcal G,\star}=0$. Thus
$\hat{\mathbf p}^k$ converges in expectation to the unique fixed point
$\mathbf p^{(0)}$ of $F^{(0)}$.

If $\mathbf p^{(0)}\in\mathrm{int}(\mathcal P)$, then the fixed point satisfies
the CV$(0)$ first-order conditions. Since $A=0$, Equation~\eqref{eq:CV-FOC}
reduces to
\[
\lambda_i(\mathbf p)+p_i\partial_{p_i}\lambda_i(\mathbf p)=0,
\qquad \forall i\in[n].
\]
These are exactly the first-order conditions for the static pricing game in
which each seller chooses $p_i$ taking rivals' prices as fixed. By
Assumption~\ref{assu:cv-foc-sufficiency}, applied with $A=0$, these
first-order conditions characterize best replies. Therefore
$\mathbf p^{(0)}$ is a Nash equilibrium.
\end{proof}

\section{Proofs and Examples for Section~\ref{sec:supra_competitive}}

We now present the examples discussed in Section~\ref{sec:supra_competitive}, showing that more feedback or more frequent experimentation can hurt a seller's limiting revenue, and that positive experimentation correlations can induce negative conjectures and hence prices below the Nash benchmark. We then prove Proposition \ref{Prop:Price_increase}. 

\begin{example}[Feedback and experimentation frequency can help or hurt]
\label{ex:feedback-frequency-can-help-or-hurt}
Consider $\lambda_1(p_1,p_2)=100-10p_1+bp_2$ and
$\lambda_2(p_1,p_2)=100-10p_2+bp_1$, where $b>0$. Seller \(2\) observes seller
\(1\)'s price. Suppose first that seller \(1\) does not observe seller \(2\)'s
price, and let $a:=A_{12}^{\mathcal G,\star}>0$, while
$A_{21}^{\mathcal G,\star}=0$. 

The  CV first-order conditions are
$100+bp_2-(20-ba)p_1=0$ and $100+bp_1-20p_2=0$. Solving gives
$p_1(a;b)=(2000+100b)/(400-b^2-20ba)$ and
$p_2(a;b)=(100+bp_1(a;b))/20$. At this equilibrium,
$\lambda_1(p(a;b))=(10-ba)p_1(a;b)$, so seller \(1\)'s limiting revenue is
$r_1(a;b)=p_1(a;b)^2(10-ba)$. Differentiating,
\[
\frac{\partial r_1(a;b)}{\partial a}
=
\frac{(2000+100b)^2b^2(b-20a)}
{(400-b^2-20ba)^3}.
\]
Thus, in the parameter region considered below, \(r_1(a;b)\) is increasing in
\(a\) for $a<\min\{b/20,\frac{400-b^2}{20b}\}$ or $a>\max\{b/20,\frac{400-b^2}{20b}\}$ and decreasing elsewhere. Hence, a stronger
positive induced coefficient can either raise or lower seller \(1\)'s limiting revenue. 

\medskip
\noindent\textbf{Changing feedback.}
A positive coefficient \(a\) can be generated, for example, by a two-point
design with $\mathbb P(Z_1=1)=1/2$,
$\mathbb P(Z_2=1\mid Z_1=1)=1/2+a/2$, and
$\mathbb P(Z_2=1\mid Z_1=0)=1/2-a/2$. Then
$A_{12}^{\mathcal G,\star}
=\operatorname{Cov}(Z_1,Z_2)/\operatorname{Var}(Z_1)
=\mathbb E[Z_2\mid Z_1=1]-\mathbb E[Z_2\mid Z_1=0]=a$.
If seller \(1\) also observes seller \(2\)'s price, then no rival price is
omitted, and the induced matrix becomes \(A^{\mathcal G,\star}=0\). Thus the
limiting equilibrium changes from the CV equilibrium indexed by \(a\) to the
Nash equilibrium.

For $b=9$ and $a=2/5$,
$r_1(2/5;9)=53824000/61009\approx882.230$, while
$r_1(0;9)=100000/121\approx826.446$. Thus observing seller \(2\)'s price lowers
seller \(1\)'s  revenue. In this case, removing the omitted-variable
term moves the price away from the privately favorable CV outcome.

For $b=4$ and $a=1/2$,
$r_1(1/2;4)=720000/1849\approx389.400$, while
$r_1(0;4)=3125/8=390.625$. Thus observing seller \(2\)'s price raises seller
\(1\)'s revenue. In this case, the positive induced coefficient is too
large: it moves seller \(1\)'s price above the revenue-improving range,
and feedback removes this harmful omitted-variable term.

\medskip
\noindent\textbf{Changing experimentation frequency.}
The same non-monotonicity can arise from changing experimentation frequency.
Assume the same feedback structure, so seller \(1\) does not observe seller
\(2\)'s price. Let $a\in(0,1/2)$ be seller \(1\)'s treatment probability,
$\mathbb P(Z_1=1)=a$, and suppose
$\mathbb P(Z_2=1\mid Z_1=1)=1/2+a(1-a)$ and
$\mathbb P(Z_2=1\mid Z_1=0)=1/2-a^2$. Then $\mathbb P(Z_2=1)=1/2$, so seller
\(2\)'s marginal experimentation frequency is fixed, while
$A_{12}^{\mathcal G,\star}
=\operatorname{Cov}(Z_1,Z_2)/\operatorname{Var}(Z_1)
=\mathbb P(Z_2=1\mid Z_1=1)-\mathbb P(Z_2=1\mid Z_1=0)=a$.
Thus increasing seller \(1\)'s treatment probability increases the positive
coefficient induced by correlated experimentation.

Now we can compare $a=1/5$ and $a=2/5$. If $b=9$, then both values are below
$\min\{b/20,\frac{400-b^2}{20b}\}$, and increasing seller \(1\)'s
experimentation frequency raises her  revenue. If instead $b=4$, then
the same increase lowers seller
\(1\)'s  revenue.

Also we note that the relevant linear-demand
convergence conditions discussed in Section \ref{sec:examples}  hold in all the numerical cases.

The example is stylized, but the mechanism is quite general. Feedback changes the
set of omitted prices and therefore changes the induced conjecture matrix.
Experimentation frequency can also change the induced conjecture matrix by
changing the correlation between observed and omitted experimentation. These
changes can be harmful or beneficial for a seller depending on the demand structure and other sellers' behavior. 
\end{example}

\begin{example}[Positive pairwise correlation and negative induced conjectures] \label{Example:NegativeConj}
This example illustrates that, under partial feedback, positive pairwise
correlation in experimentation need not imply positive induced conjectures.
Consider four sellers. Sellers \(1,2,3\) are smaller sellers, and seller \(4\)
is a technologically sophisticated major seller. The smaller sellers monitor the
major seller but not each other, while the major seller tracks all smaller
sellers:
\[
\mathcal S_i=\{4\},\quad i=1,2,3,
\qquad
\mathcal S_4=\{1,2,3\}.
\]
Thus, for \(i=1,2,3\), \(\mathcal R_i=\{i,4\}\) and
\(\mathcal U_i=\{1,2,3\}\setminus\{i\}\), while \(\mathcal U_4=\emptyset\).

Let \(\rho=\sqrt{0.2}\), and suppose the limiting covariance matrix of
\(Z=(Z_1,Z_2,Z_3,Z_4)\) is
\[
\Sigma_Z^\star=
\begin{pmatrix}
1 & 0.1 & 0.1 & \rho\\
0.1 & 1 & 0.1 & \rho\\
0.1 & 0.1 & 1 & \rho\\
\rho & \rho & \rho & 1
\end{pmatrix}.
\]
All pairwise correlations are positive. The matrix is positive definite: two
eigenvalues are \(0.9\), and the remaining two are the eigenvalues of
\(\begin{pmatrix}1.2 & \sqrt{3}\rho\\ \sqrt{3}\rho & 1\end{pmatrix}\), whose
determinant is \(1.2-3\rho^2=0.6>0\). This covariance structure can be generated
by bounded experimentation, e.g., by taking \(Z=L\xi\), where
\(LL^\top=\Sigma_Z^\star\) and \(\xi\) has independent Rademacher coordinates.

Now fix a smaller seller \(i\in\{1,2,3\}\) and an omitted smaller seller
\(j\in\{1,2,3\}\setminus\{i\}\). With seller \(i\)'s observed regressors ordered
as \((Z_i,Z_4)\), we have
\[
\Sigma_i^\star=
\begin{pmatrix}
1 & \rho\\
\rho & 1
\end{pmatrix},
\qquad
\mathbf c_{ij}^\star=(0.1,\rho)^\top .
\]
Hence
\[
[\pi_{ij}^\star]_i
=
\left[(\Sigma_i^\star)^{-1}\mathbf c_{ij}^\star\right]_i
=
\frac{0.1-\rho^2}{1-\rho^2}
=
\frac{0.1-0.2}{0.8}
=
-\frac18 .
\]
Thus \(A_{ij}^{\mathcal G,\star}=-1/8\) for every
\(i\in\{1,2,3\}\) and every omitted
\(j\in\{1,2,3\}\setminus\{i\}\). Entries corresponding to observed rivals are
zero by definition, and seller \(4\) observes all rivals. Therefore
\[
A^{\mathcal G,\star}
=
\begin{pmatrix}
0 & -1/8 & -1/8 & 0\\
-1/8 & 0 & -1/8 & 0\\
-1/8 & -1/8 & 0 & 0\\
0 & 0 & 0 & 0
\end{pmatrix}.
\]
Thus all nonzero induced conjectures are negative, even though every pair of
sellers' experimentation variables is positively correlated. Intuitively, the major seller creates a common source of co-movement in
experimentation; after linearly projecting out the major seller's price
variation, the remaining co-movement between smaller sellers is negative.
\end{example}

\begin{proof}[Proof of Proposition~\ref{Prop:Price_increase}]
\textbf{Step 1.}  For each $i$ and every $p_i\in[p_i^l,p_i^h]$, if $\mathbf{p}_{-i}\le \mathbf{p}'_{-i}$ and $A\preceq A'$, then $G_i(p_i,\mathbf{p}_{-i};A)\ \le\ G_i(p_i,\mathbf{p}'_{-i};A')$. 

\textbf{Proof of Step 1.} Let $A,A'\in\mathcal A$ with $A\preceq A'$. For any $\mathbf{p}\in\mathcal P$,
\[
G_i(\mathbf{p};A')-G_i(\mathbf{p};A)
=
p_i\sum_{j\neq i}(A'_{ij}-A_{ij})\,\partial_{p_j}\lambda_i(\mathbf{p})
\ \ge\ 0,
\]
since $p_i\ge 0$, $A'_{ij}-A_{ij}\ge 0$, and $\partial_{p_j}\lambda_i(\mathbf{p})\ge 0$ for all $j\neq i$.

In addition, since $G_i(\cdot;A)$ is $C^1$ on the rectangle $\mathcal P$ and
$\partial_{p_j}G_i(\mathbf{p};A)\ge 0$ for all $j\neq i$, it follows that
$\mathbf{p}_{-i}\le \mathbf{p}'_{-i}\Rightarrow G_i(p_i,\mathbf{p}_{-i};A)\le G_i(p_i,\mathbf{p}'_{-i};A)$.

It follows that for $\mathbf{p}_{-i}\le \mathbf{p}'_{-i}$ and $A\preceq A'$,
we have 
$G_i(p_i,\mathbf{p}_{-i};A)\ \le\ G_i(p_i,\mathbf{p}_{-i};A')\ \le\ G_i(p_i,\mathbf{p}'_{-i};A')$
which proves Step 1.

\textbf{Step 2.}
Define the projection $\Pi:\mathbb R^n\to\mathcal P$ coordinatewise by
$\Pi_i(z):=\min\{ p_i^h,\max\{p_i^l,z\}\}$ and define the self-map
$f:\mathcal P\times\mathcal A\to\mathcal P$ by
\[
f(\mathbf{p},A)\ :=\ \Pi\big(\mathbf{p}+ G(\mathbf{p};A)\big).
\]
By Step~1 and monotonicity of $\Pi$, $(\mathbf{p}_{-i},A)\mapsto f_i(p_i,\mathbf{p}_{-i};A)$ is nondecreasing.
By Theorem~4 in \cite{milgrom1994comparing}, the lowest and highest fixed points of $f(\cdot,A)$,
denoted $\mathbf{p}_L(A)$ and $\mathbf{p}_H(A)$, are coordinatewise nondecreasing in $A$.\footnote{Theorem~4 in \cite{milgrom1994comparing} is stated for self-maps on $[0,1]^n$.
Define the coordinatewise affine, order-preserving bijection
$S:\mathcal P\to[0,1]^n$ by
$
S_i(p_i):=\frac{p_i- p_i^l}{p_i^h- p_i^l}
$
(with $p_i^h> p_i^l$), and define $\tilde f(x,A):=S(f(S^{-1}(x),A))$.
Fixed points and the least/greatest order are preserved under $S$.}

Now let $\mathbf p$ be any fixed point of $f(\cdot,A)$. If $p_i= p_i^l$, then
\[
f_i(\mathbf{p},A)=\Pi_i\big(p_i^l+G_i(p_i^l,\mathbf{p}_{-i};A)\big)> p_i^l,
\]
contradicting $f_i(\mathbf{p},A)=p_i$. Similarly, if $p_i=p_i^h$ then $f_i(\mathbf{p},A)<p_i^h$, also a contradiction.
Hence, every fixed point lies in $\mathrm{int}(\mathcal P)$.

Now let $\mathbf{p}\in\mathrm{int}(\mathcal P)$ be a fixed point. Since $p_i\in(p_i^l,p_i^h)$,
the projection cannot bind at a fixed point, so $\mathbf{p}=f(\mathbf{p},A)=\mathbf{p}+G(\mathbf{p};A)$, implying $G(\mathbf{p};A)=0$.
Conversely, if $\mathbf{p}\in\mathrm{int}(\mathcal P)$ satisfies $G(\mathbf{p};A)=0$, then $f(\mathbf{p},A)=\Pi(\mathbf{p}+0)=\mathbf{p}$.

Therefore, fixed points of $f(\cdot,A)$ coincide with interior solutions to $G(\mathbf{p};A)=0$.
In particular, the extremal (lowest/highest) CV equilibria are coordinatewise nondecreasing in $A$.
If the CV equilibrium is unique, then $\mathbf{p}(A)$ is coordinatewise nondecreasing in $A$.
\end{proof}

\section{Demand Examples: Proofs from Section \ref{sec:examples} }

\begin{lemma}\label{lem:jacobian-z-linear-general}
Consider the linear demand model \eqref{eq:linear_model} with $b_{ii}>0$ and
$b_{ij}\ge0$, and fix a conjecture matrix $A$ with
$\beta_i^{(A)}=b_{ii}-\sum_{j\ne i}A_{ij}b_{ij}>0$ for every seller $i$. Then
$z^{(A)}$ has a constant Jacobian and
\[
\|Dz^{(A)}\|_\infty
=
\max_{i\in[n]}
\frac{
\left|\sum_{j\ne i}A_{ij}b_{ij}\right|+\sum_{j\ne i}b_{ij}
}{
2\,\beta_i^{(A)}
}.
\]
In the case where conjectures are nonnegative, $A_{ij}\ge0$, this
simplifies to
\[
\|Dz^{(A)}\|_\infty
=
\max_{i\in[n]}
\frac{
\sum_{j\ne i}(1+A_{ij})b_{ij}
}{
2\left(b_{ii}-\sum_{j\ne i}A_{ij}b_{ij}\right)
}.
\]
Consequently, in this case, $\|Dz^{(A)}\|_\infty<1$ is equivalent to
\[
\sum_{j\ne i}(1+3A_{ij})b_{ij}<2b_{ii},
\qquad i\in[n].
\]
\end{lemma}

\begin{proof}[Proof of Lemma~\ref{lem:jacobian-z-linear-general}]
For linear demand,
$\partial_{p_i}\lambda_i=-b_{ii}$ and $\partial_{p_j}\lambda_i=b_{ij}$ for
$j\ne i$. Hence
$\beta_i^{(A)}=b_{ii}-\sum_{j\ne i}A_{ij}b_{ij}$ is constant. Since
$z_i^{(A)}(\mathbf p)=p_i/2+\lambda_i(\mathbf p)/(2\beta_i^{(A)})$,
\[
\frac{\partial z_i^{(A)}}{\partial p_i}
=
\frac12-\frac{b_{ii}}{2\beta_i^{(A)}}
=
-\frac{\sum_{j\ne i}A_{ij}b_{ij}}{2\beta_i^{(A)}},
\qquad
\frac{\partial z_i^{(A)}}{\partial p_j}
=
\frac{b_{ij}}{2\beta_i^{(A)}}\quad (j\ne i).
\]
The induced infinity norm is the maximum absolute row sum, so the $i$th row sum
is
\[
\left|
\frac{-\sum_{j\ne i}A_{ij}b_{ij}}{2\beta_i^{(A)}}
\right|
+
\sum_{j\ne i}
\frac{b_{ij}}{2\beta_i^{(A)}} ,
\]
which gives the first display. If $A_{ij}\ge0$ and $b_{ij}\ge0$, the absolute
value signs can be removed. The condition $\|Dz^{(A)}\|_\infty<1$ then becomes
$\sum_{j\ne i}(1+A_{ij})b_{ij}
<2(b_{ii}-\sum_{j\ne i}A_{ij}b_{ij})$, which is equivalent to the stated
condition.
\end{proof}

\begin{lemma}\label{lem:mnl_contraction}
Consider the MNL demand model \eqref{eq:mnl_linear}. Fix a conjecture matrix
$A$ and define
\[
h_i(\mathbf p)
:=
b_i(1-\lambda_i(\mathbf p))
-
\sum_{j\ne i}A_{ij}b_j\lambda_j(\mathbf p).
\]
If $h_i(\mathbf p)>0$ on $\mathcal P$, then
$\beta_i^{(A)}(\mathbf p)=\lambda_i(\mathbf p)h_i(\mathbf p)>0$ and
\[
z_i^{(A)}(\mathbf p)=\frac12p_i+\frac{1}{2h_i(\mathbf p)}.
\]
Moreover,
\[
\frac{\partial z_i^{(A)}}{\partial p_i}(\mathbf p)
=
\frac12-\frac{b_i\lambda_i(\mathbf p)}{2h_i(\mathbf p)},
\]
and, for $m\ne i$,
\[
\frac{\partial z_i^{(A)}}{\partial p_m}(\mathbf p)
=
\frac{b_m\lambda_m(\mathbf p)}
{2h_i(\mathbf p)^2}
\left[
b_i\lambda_i(\mathbf p)
+
\sum_{\substack{j\ne i\\ j\ne m}}
A_{ij}b_j\lambda_j(\mathbf p)
-
A_{im}b_m(1-\lambda_m(\mathbf p))
\right].
\]
In particular, if $A=0$ and $b_i=b$ for all $i$, then
$\sup_{\mathbf p\in\mathcal P}\lambda_i(\mathbf p)<3/5$ for every seller $i$
implies $\sup_{\mathbf p\in\mathcal P}\|Dz^{(0)}(\mathbf p)\|_\infty<1$.
\end{lemma}

\begin{proof}[Proof of Lemma~\ref{lem:mnl_contraction}]
For the MNL model,
\[
\partial_{p_i}\lambda_i=-b_i\lambda_i(1-\lambda_i),
\qquad
\partial_{p_j}\lambda_i=b_j\lambda_i\lambda_j\quad (j\ne i).
\]
Therefore
\[
\beta_i^{(A)}
=
-\left(\partial_{p_i}\lambda_i+\sum_{j\ne i}A_{ij}\partial_{p_j}\lambda_i\right)
=
\lambda_i
\left(
b_i(1-\lambda_i)-\sum_{j\ne i}A_{ij}b_j\lambda_j
\right)
=
\lambda_i h_i .
\]
Since $\lambda_i>0$, $h_i>0$ implies $\beta_i^{(A)}>0$, and
$z_i^{(A)}=p_i/2+\lambda_i/(2\beta_i^{(A)})=p_i/2+1/(2h_i)$.

It remains to differentiate $h_i$. For the own-price derivative,
\[
\partial_{p_i}h_i
=
-b_i\partial_{p_i}\lambda_i
-
\sum_{j\ne i}A_{ij}b_j\partial_{p_i}\lambda_j
=
b_i\lambda_i
\left(
b_i(1-\lambda_i)-\sum_{j\ne i}A_{ij}b_j\lambda_j
\right)
=
b_i\lambda_i h_i.
\]
Hence
$\partial_{p_i}z_i^{(A)}
=1/2-(\partial_{p_i}h_i)/(2h_i^2)
=1/2-b_i\lambda_i/(2h_i)$.

For $m\ne i$,
\[
\begin{aligned}
\partial_{p_m}h_i
&=
-b_i\partial_{p_m}\lambda_i
-
\sum_{j\ne i}A_{ij}b_j\partial_{p_m}\lambda_j  \\
&=
-b_ib_m\lambda_i\lambda_m
+
A_{im}b_m^2\lambda_m(1-\lambda_m)
-
\sum_{\substack{j\ne i\\ j\ne m}}
A_{ij}b_jb_m\lambda_j\lambda_m  \\
&=
b_m\lambda_m
\left[
A_{im}b_m(1-\lambda_m)
-
b_i\lambda_i
-
\sum_{\substack{j\ne i\\ j\ne m}}A_{ij}b_j\lambda_j
\right].
\end{aligned}
\]
Since $\partial_{p_m}z_i^{(A)}=-(\partial_{p_m}h_i)/(2h_i^2)$, the stated
off-diagonal formula follows.

Now set $A=0$ and $b_i=b$ for all $i$. Then
$h_i=b(1-\lambda_i)$,
\[
\frac{\partial z_i^{(0)}}{\partial p_i}
=
\frac{1-2\lambda_i}{2(1-\lambda_i)},
\qquad
\frac{\partial z_i^{(0)}}{\partial p_j}
=
\frac{\lambda_i\lambda_j}{2(1-\lambda_i)^2}\quad (j\ne i).
\]
Thus the $i$th absolute row sum satisfies
\[
\sum_j
\left|
\frac{\partial z_i^{(0)}}{\partial p_j}
\right|
=
\frac{|1-2\lambda_i|}{2(1-\lambda_i)}
+
\frac{\lambda_i}{2(1-\lambda_i)^2}
\sum_{j\ne i}\lambda_j
<
\frac{|1-2\lambda_i|+\lambda_i}{2(1-\lambda_i)},
\]
where we used $\sum_{j\ne i}\lambda_j<1-\lambda_i$. If $\lambda_i\le1/2$, the
last expression equals $1/2$. If $\lambda_i>1/2$, it equals
$(3\lambda_i-1)/(2(1-\lambda_i))$, which is less than one if and only if
$\lambda_i<3/5$. Taking the supremum over $\mathbf p\in\mathcal P$ proves the
claim.
\end{proof}

\begin{lemma}\label{lm:best_reply_satisfying_models}
Fix a conjecture matrix $A$. For the linear demand model \eqref{eq:linear_model},
if $\beta_i^{(A)}>0$ for every seller $i$, then
Assumption~\ref{assu:cv-foc-sufficiency} holds. For the MNL demand model
\eqref{eq:mnl_linear}, Assumption~\ref{assu:cv-foc-sufficiency} holds for any
conjecture matrix $A$.
\end{lemma}

\begin{proof}[Proof of Lemma~\ref{lm:best_reply_satisfying_models}]
Fix a seller $i$, a conjecture row $A_{i\cdot}$, and a price profile
$\mathbf p\in\mathcal P$. Let
$v_i(A_{i\cdot})=e_i+\sum_{j\ne i}A_{ij}e_j$ and
$\mathbf p(s):=\mathbf p+s\,v_i(A_{i\cdot})$. Let
$S_i(\mathbf p;A):=\{s\in\mathbb R:\mathbf p(s)\in\mathcal P\}$, and consider
the one-dimensional revenue
\[
g_i(s):=(p_i+s)\lambda_i(\mathbf p(s)),
\qquad s\in S_i(\mathbf p;A).
\]
By the chain rule, $g_i'(0)=\nabla r_i(\mathbf p)^\top v_i(A_{i\cdot})$, so
the CV first-order condition is exactly the stationarity condition for this
one-dimensional problem.

For the linear model, along $\mathbf p(s)$,
\[
\lambda_i(\mathbf p(s))
=
\lambda_i(\mathbf p)
+
s\left(-b_{ii}+\sum_{j\ne i}A_{ij}b_{ij}\right),
\]
and hence
\[
g_i''(s)
=
2\left(-b_{ii}+\sum_{j\ne i}A_{ij}b_{ij}\right)
=
-2\beta_i^{(A)}<0.
\]
Thus $g_i$ is strictly concave on the feasible interval, so any interior
stationary point is the unique global maximizer.

For the MNL model, $\lambda_i(\mathbf p)>0$ and $p_i>0$ on $\mathcal P$, so
maximizing $g_i$ is equivalent to maximizing $\log g_i$. Write
$f_j(p_j):=a_j-b_jp_j$ and
$S(\mathbf p):=\log(1+\sum_j e^{f_j(p_j)})$. Then
\[
\log g_i(s)
=
\log(p_i+s)+f_i(p_i+s)-S(\mathbf p(s)).
\]
The first term has second derivative $-1/(p_i+s)^2$, the second term is linear
in $s$, and the last term is the negative of a convex log-sum-exp function
composed with an affine path. Therefore $\log g_i(s)$ is strictly concave on
$S_i(\mathbf p;A)$. Hence any interior stationary point is the unique global
maximizer, verifying Assumption~\ref{assu:cv-foc-sufficiency}.
\end{proof}

\end{document}